\documentclass[12pt]{article}
\usepackage[small]{titlesec}
\usepackage{natbib,amsmath,amssymb,amsthm,graphicx,setspace,paralist,booktabs,rotating,subcaption,float,color}
\usepackage{graphicx}
\usepackage{enumerate}
\usepackage{url} 
\usepackage{amssymb, amsmath, amsthm}
\usepackage{breakcites}

\usepackage{bigstrut} 
\usepackage{blkarray}
\usepackage{booktabs}
\usepackage{enumitem} 
\usepackage{float} 
\usepackage{graphicx}
\usepackage{subcaption}
\usepackage{mwe}
\usepackage{comment}
\usepackage[colorlinks=true, allcolors=LinkBlue, backref=page]{hyperref}
\usepackage{lmodern}
\usepackage{mathrsfs} 
\usepackage{multirow}
\usepackage{natbib} 
\usepackage{scalerel}
\usepackage{setspace}
\usepackage{soul} 
\usepackage{subcaption}
\usepackage{titling}
\usepackage{verbatim}
\usepackage{xcolor}
\usepackage{xfrac}
\usepackage{algpseudocode}
\usepackage{bbm}
\usepackage{setspace}
\usepackage{amsfonts}
\usepackage{outlines}

\AtBeginDocument{
	\abovedisplayskip=5pt plus 2pt minus 2pt
	\belowdisplayskip=\abovedisplayskip
	\abovedisplayshortskip=2pt plus 2pt minus 2pt
	\belowdisplayshortskip=\belowdisplayskip}

\titlespacing*{\section}{0pt}{*2}{*1}
\titlespacing*{\subsection}{0pt}{*2}{*1} 
\bibpunct{(}{)}{;}{a}{}{,}
\setlist{noitemsep, topsep=0pt} 
\definecolor{LinkBlue}{rgb}{.15, .25, .85} 

\binoppenalty=\maxdimen 
\relpenalty=\maxdimen 

\makeatletter
\renewcommand*{\NAT@spacechar}{~}
\makeatother

\usepackage{algorithm}
\makeatletter

\newcommand{\blind}{1}

\newfloat{algorithm}{tbp}{loa}
\providecommand{\algorithmname}{Algorithm}
\floatname{algorithm}{\protect\algorithmname}

\newtheorem{theorem}{Theorem}
\newtheorem{example}{Example}
\newtheorem{lemma}{Lemma}

\newtheorem{corollary}{Corollary}

\newtheorem{assumption}{Assumption}[section]
\newtheorem{prop}{Proposition}

\def \bP {\mathbb{P}}
\def \bE {\mathbb{E}}
\def \bR {\mathbb{R}}

\def \var {\mathsf{Var}}

\usepackage{xspace}

\newcommand{\TV}{{\sf TV}}

\newcommand{\diff}{\mathrm{d}}

\newcommand{\Poi}{\mathsf{Poi}}
\newcommand{\Geo}{{\mathsf{Geo}}}
\newcommand{\Binom}{{\mathsf{Binom}}}

\newcommand{\Indc}{\mathbf{1}}

\newcommand{\sfE}{{\mathsf{E}}}

\newcommand{\sfN}{{\mathsf{N}}}
\newcommand{\sfO}{{\mathsf{O}}}

\newcommand{\sfU}{{\mathsf{U}}}

\newcommand{\calC}{{\mathcal{C}}}
\newcommand{\calD}{{\mathcal{D}}}

\newcommand{\calF}{{\mathcal{F}}}

\newcommand{\calO}{{\mathcal{O}}}

\newcommand{\calS}{{\mathcal{S}}}
\newcommand{\calT}{{\mathcal{T}}}

\newcommand{\barP}{{\bar{P}}}

\begin{document}

\def\spacingset#1{\renewcommand{\baselinestretch}%
{#1}\small\normalsize} \spacingset{1}


\if1\blind
{
  \title{\bf Unbiased Multilevel Monte Carlo methods for intractable distributions: MLMC meets MCMC}
  \author{Tianze Wang,\ Guanyang Wang}
  \maketitle
  	\footnotetext{Guanyang Wang is Assistant Professor, Department of Statistics, Rutgers University, Piscataway, NJ 08854. (E-mail: \emph{guanyang.wang@rutgers.edu}). Tianze Wang is Ph.D. candidate, Department of Statistics, Rutgers University, Piscataway, NJ 08854. (E-mail: \emph{tianze.wang@rutgers.edu}) .}
  
} \fi

\if0\blind
{
  \bigskip
  \bigskip
  \bigskip
  \begin{center}
    {\bf Unbiased Multilevel Monte Carlo methods for intractable distributions: MLMC meets MCMC}
\end{center}
  \medskip
} \fi

\bigskip
\begin{abstract}
Constructing unbiased estimators from Markov chain Monte Carlo (MCMC) outputs is a difficult problem that has recently received a lot of attention in the statistics and machine learning communities. However, the current unbiased MCMC framework only works when the quantity of interest is an expectation, which excludes many practical applications. In this paper, we propose a general method for constructing unbiased estimators for functions of expectations and extend it to construct unbiased estimators for nested expectations. Our approach combines and generalizes the unbiased MCMC and Multilevel Monte Carlo (MLMC) methods. In contrast to traditional sequential methods, our estimator can be implemented on parallel processors. We show that our estimator has a finite variance and computational complexity and can achieve $\varepsilon$-accuracy within the optimal $O(1/\varepsilon^2)$ computational cost under mild conditions. Our numerical experiments confirm our theoretical findings and demonstrate the benefits of unbiased estimators in the massively parallel regime.

\end{abstract}

\noindent%
\vfill

\newpage
\spacingset{1.9} 
\section{Introduction}\label{sec:intro}
Monte Carlo methods generate unbiased estimators for the expectation of a distribution.  In practice, however, it may be impractical to sample from the underlying distribution and the quantity of interest may not be an expectation. Generally, most inference problems can be represented as estimating a quantity of the form $\calT(\pi)$, where $\pi$ is one or a group of distributions and $\calT$ is a functional of $\pi$. We begin by considering several motivating examples to gain a deeper understanding of the different forms that $\calT(\pi)$ might take.
\begin{example}[Integration]\label{eg:integration}
	Let $\pi$ be a probability distribution and $f$ a $\pi$-integrable function. The problem of estimating $\bE_\pi[f]$ can be viewed as estimating $\calT(\pi)$ where $\calT$ is the integral operator: $\calT(\pi) := \int f(x)\pi(dx).$

\end{example}

\begin{example}[Nested Monte Carlo]\label{eg: nested}
	Let $\pi$ be a probability distribution, and suppose the quantity of our interest has the form $\calT(\pi) := \bE_\pi[\lambda]$, where $\lambda$ is itself intractable. The intractable function $\lambda$ may take the form $\lambda(x) := f(x, \gamma(x)),$ where $\gamma(x) = \bE_{y\sim p(y\mid x)}[\phi(x,y)]$ is a conditional expectation. One concrete example is the two-stage optimal stopping problem,  where $\gamma(x) = \max\{x, \bE[y|x]\}$. Estimating the nested expectation is known as a challenging problem in Monte Carlo methods due to its involved structure \citep{rainforth2018nesting}. 
\end{example}

\begin{example}[Ratios of normalizing constants]\label{eg: ratio of normalizing constant}
	Let $\pi_1(x) =f_1(x)/Z_1$ and  $\pi_2(x) = f_2(x)/Z_2$ be two probability densities with common support. We assume $f_1$ and $f_2$ can be easily evaluated, but the normalizing constants $Z_1$ and $Z_2$ are computationally intractable. Consider the task of estimating the ratio of normalizing constants, i.e., $Z_1/Z_2$, standard calculation shows $Z_1/Z_2 = \bE_{\pi_2}[f_1]/\bE_{\pi_1}[f_2].$
	The problem can be viewed as estimating $\calT(\pi)$ by choosing $\pi$ as the product measure $\pi_1\times\pi_2$, and 
	$\calT(\pi) := \bE_{\pi_2}[f_1]/\bE_{\pi_1}[f_2].$
	The problem finds statistical and physics applications, including hypothesis testing, Bayesian inference, and estimating free energy differences. We refer the readers to \cite{meng1996simulating} for other applications.
\end{example}

\begin{example}[Quantile estimation]\label{eg: quantile}
	Let $\pi$ be a probability distribution with  cumulative distribution function $F_\pi$ and $q$ a constant in $(0,1)$. Estimating the $q$-th quantile of $\pi$ can be formulated as estimating $\calT(\pi)$ where $
	\calT(\pi) := \inf_v\{F_\pi(v)\geq q\}. $
	Quantile estimation problem has applications in statistics, economics, and other fields. We refer the readers to \cite{koenker2001quantile, takeuchi2006nonparametric, romano2019conformalized} for more discussions, and \cite{doss2014markov} for an MCMC-based method.
\end{example}

In all the examples above, the distribution $\pi$ can be intractable.  In some cases, such as Example \ref{eg:integration} and \ref{eg: nested}, the quantity of interest is an expectation under $\pi$, although the function inside the expectation may or may not be intractable. In other cases, including Example \ref{eg: ratio of normalizing constant} and \ref{eg: quantile},  $\calT$ is a functional of $\pi$, but not an expectation.

Throughout this paper, we focus on designing unbiased estimators of $\calT(\pi)$ assuming one can only access outputs from some MCMC algorithm that leaves $\pi$ as stationary distribution.  Unbiased estimators are of particular interest because they can help users save computation time in a parallel implementation environment.  To elaborate, classical MCMC estimators, which are based on the empirical distribution after running the MCMC algorithm for a fixed number of iterations, are generally biased unless the algorithm is initialized at the target distribution $\pi$.  This bias can be problematic in a parallel computing environment, where the number of processors is huge but the computational budget per processor is limited. In contrast, unbiased estimators can be computed on different devices in parallel without communication, allowing users to control the mean-squared error (which is only determined by the variance) to an arbitrarily low level by simply increasing the number of processors.  Evidences support the advantage of unbiased estimators in parallel Monte Carlo algorithms are provided in \cite{rosenthal2000parallel,nguyen2022many}.

On top of parallel computing, the confidence intervals can be easily constructed using unbiased estimators from Monte Carlo outputs to improve uncertainty quantification in cases where the variance is hard to estimate. Moreover, these unbiased estimators are often more adaptable and can be used as subroutines in more complicated Monte Carlo problems like pseudo-marginal MCMC algorithms \citep{andrieu2009pseudo} and nested Monte Carlo problems \citep{rainforth2018nesting,zhou2021unbiased}.

Without further assumption on $\calT$ and $\pi$, it is well known that constructing unbiased estimators of $\calT(\pi)$ is difficult. Computational challenges appear in both components of the pair $(\calT, \pi)$. The bias of standard Monte Carlo estimators arises from the nonlinearity of $\calT$ and the sampling error of the MCMC algorithm. Fortunately, recent works provide promising solutions when one component of the above $(\calT, \pi)$ pair is easy while the other is relatively difficult. We briefly review the following two cases separately:
\begin{itemize}
	\item(Case 1: Easy $\calT$, difficult $\pi$): When $\calT$ is an integral operator with respect to some tractable function $f$, but $\pi$ is infeasible to sample from, i.e., $\calT(\pi) := \bE_\pi[f]$ for some intractable $\pi$. The problem is considered by Jacob, O'Leary, and Atchad{\'e} (JOA henceforth) \citep{jacob2020unbiased}. The JOA estimator,  which follows the idea of \cite{glynn2014exact}, solves this problem via couplings of Markov chains. The unbiased MCMC framework has recently raised much attention. It has been applied in convergence diagnostics \citep{biswas2019estimating, biswas2021bounding, biswas2020coupling}, gradient estimation \citep{ruiz2020unbiased}, asymptotic variance estimation \cite{douc2022solving}, and so on. 
	\item(Case 2: Easy $\pi$, difficult $\calT$): When $\pi$ can be sampled perfectly, but $\calT(\pi) := g(\bE_\pi[f])$ is a function of the expectation, or $\calT$ is an expectation with respect to a function which further depends on an expectation (e.g, the nested expectation), the state of the art debiasing technique is the unbias MLMC  method developed by McLeish, Glynn, Rhee, and  Blanchet \citep{blanchet2015unbiased, rhee2015unbiased, mcleish2011general} which is a randomized version of the celebrated (non-randomized) MLMC methods pioneered by Heinrich and Giles  \citep{heinrich2001multilevel, giles2008multilevel, giles2015multilevel}. Unbiased MLMC methods have also found many applications, including gradient estimation \citep{shi2021multilevel}, optimal stopping \citep{zhou2021unbiased}, robust optimization \citep{levy2020large}. 
\end{itemize}

In summary, the unbiased MCMC method assumes easy $\cal T$ (an integral operator) but difficult $\pi$, and the unbiased MLMC method assumes easy $\pi$ (perfectly simulable) but difficult $\cal T$. Both assumptions can be violated in many practical applications, such as Example \ref{eg: nested} -- \ref{eg: quantile}. Although immense progress has been made, there is no systematic way of constructing unbiased estimators for general $\calT(\pi)$ beyond special cases.

In this article, we present a step toward designing unbiased estimators of $\calT(\pi)$ for the general $(\calT,\pi)$ pair by combining and extending the ideas of the unbiased MCMC and MLMC methods.  We propose generic unbiased estimators for functions of expectations, i.e., $\calT(\pi) = g(m(\pi)):= g(\bE_\pi[f(X)])$ where $\pi$ is a $d$-dimensional probability measure that can only be approximately sampled by MCMC methods, $f: \bR^d \rightarrow \bR^m$ is a deterministic map, and $g: \bR^m \rightarrow \bR$ is a deterministic function \footnote{For simplicity, we only consider scalar-valued $g$ in this paper, though our method can be naturally generalized to vector-valued functions.}. Other technical assumptions will be made clear in the subsequent sections. The unbiased estimator is easily parallelizable. It has both finite variance and computational cost for a general class of problems, which implies a `square root convergence rate' that matches the optimal rate of Monte Carlo methods \citep{novak2006deterministic} given by the Central Limit Theorem. Moreover, some technical assumptions on $g$ relax the standard `linear growth' assumption  in \cite{Blanchet2015UnbiasedMC} and \cite{blanchet2019unbiased}, which may be of independent interest. 

Our method can be naturally generalized to the unbiased estimation of the nested expectation introduced in Example \ref{eg: nested} under intractable distributions. The nested expectation is commonly regarded as a challenging task for Monte Carlo simulation. Even if one can sample perfectly from the underlying distribution, the standard `plug-in' Monte Carlo estimator is not only biased but also has a suboptimal computational cost ($\calO(\epsilon^{-3})$ or $\calO(\epsilon^{-4})$) under varying assumptions to achieve a mean square error (MSE) of $\epsilon^2$.  The proposed estimator has three advantages over the standard `plug-in' estimator. It is unbiased, has  $\calO(\epsilon^{-2})$ expected computational cost to achieve $\epsilon^2$-MSE, and works when the conditional distribution can only be approximated by MCMC methods.  

Our method naturally connects the unbiased MCMC with the MLMC method. Unbiased MCMC is an emerging area in statistics and machine learning for its potential for parallelization. The methodology in \cite{jacob2020unbiased} has been extended to different MCMC algorithms, including the Hamiltonian Monte Carlo \citep{heng2019unbiased} and the pseudo-marginal   MCMC \citep{middleton2020unbiased}. In contrast, the MLMC method (both the non-randomized and randomized version) is shown to be successful in applied math, operation research, and computational finance for estimating the expectation of SDE solutions \citep{giles2008multilevel,rhee2015unbiased}, option pricing \citep{belomestny2015multilevel, zhou2021unbiased}, and inverse problems \citep{hoang2013complexity, dodwell2015hierarchical,beskos2017multilevel, jasra2018multi}. 
When the quantity of interest is $\bE_\pi[f]$ for challenging underlying distribution $\pi$ (in contrast to $g(\bE_\pi[f])$ that we considered here), there already exists similar ideas on combining the unbiased MLMC and MCMC framework on specific problems. In  \cite{heng2021unbiased_b}, \cite{heng2021unbiased_a}, the authors propose a four-way coupling mechanism to unbiasedly estimate $\bE_\pi[f]$ when $\pi$ arises from some stochastic differential equations. Nevertheless, overall, the connections between unbiased MCMC and MLMC methods still seem largely unexplored. We hope this work will serve as a bridge for these communities and invite researchers from broader areas to develop these methods together.

The rest of this paper is organized as follows. Section \ref{subsec:notation} introduces the notations. In Section \ref{sec:existing methods}, we describe the high-level idea behind our method without diving into details. This section will also clarify the connections between unbiased MCMC and MLMC methods. We formally propose our unbiased estimator in Section \ref{subsec:construction}. In Section \ref{subsec:nest}, we generalize our estimator for estimating nested expectations. 
In Section \ref{subsec:theory}, we state the assumptions and prove the theoretical properties. In Section \ref{sec:numerical}, we implement our method on several examples to study its empirical performance. We conclude this paper in Section \ref{sec:conclusion}. Technical details such as proofs and additional experiments are deferred to the Appendix.

\subsection{Notations}\label{subsec:notation}
Throughout this article, we preserve the notation $g$ to denote a function from its domain $\calD \subset \bR^m$ to $\bR$. We write $\pi$ as a $d$-dimensional probability measure, and $\pi_1, \cdots, \pi_d$ for its marginal distributions. We denote by $m_f(\pi): =  \bE_\pi[f(X)]$ the expected value/vector of $f$ under $\pi$, and write it as $m(\pi)$  when it is unlikely to cause confusion.  The $L^p$ norm of $v\in \bR^d$ is written as $\lVert v\rVert_p := \left({\sum_{i=1}^d \lvert v_i \rvert^p}\right)^{1/p}$. For the $L^2$ norm, we simply write $\lVert v\rVert := \lVert v\rVert_2$. The geometric distribution with success probability $r$ is denoted by $\Geo(r)$, and write its probability mass function  as $p_n = p_n(r)  = (1-r)^{n-1} r$. The uniform distribution on $[0,1]$ is denoted by $\sfU[0,1]$. The multivariate normal  with mean $\mu$ and covariance matrix $\Sigma$ is denoted by $\sfN(\mu,\Sigma)$. The binomial distribution with $N$ trials and parameter $p$ is denoted by $\Binom(N,p)$. The Poisson distribution with parameter $\lambda$ is denoted by $\Poi(\lambda)$. Given a set $A\subset \bR^d$, we denote by $A^\circ$ all the interior points of $A$.  For a differentiable function $h: \bR^d \rightarrow \bR$, we denote by $Dh := (\frac{\partial h}{\partial x_1}, \frac{\partial h}{\partial x_2}, \cdots, \frac{\partial h}{\partial x_d})$ the gradient of $h$. Given two probability measures $\mu$ and $\nu$, we write their total variation (TV) distance  as $\lVert \mu - \nu \rVert_\TV := \sup_A \lvert \mu(A) - \nu(A)\rvert$. We adopt the convention that $\sum_{i = m}^n a_i = 0$ if $m>n$.

\section{A Simple Identity: Unbiased MCMC meets MLMC}\label{sec:existing methods}
Consider the task of designing unbiased estimators of $g(m(\pi)) = g\big(\bE_\pi[f(X)]\big)$. The problem is extensively studied in the literature when one can draw independent and identically distributed ($i.i.d.$) samples from $\pi$. Unbiased estimators are known to exist or not exist under different contexts \citep{keane1994bernoulli,jacob2015nonnegative}. Different debiasing techniques \citep{nacu2005fast,blanchet2015unbiased,Blanchet2015UnbiasedMC, vihola2018unbiased} have been proposed and analyzed. 
Among existing methods, the unbiased MLMC framework works with the greatest generality.

When $\pi$ is infeasible to sample from, our first observation is based on the following simple identity. For  every random variable $H$ with $\bE[H] = m(\pi)$, we have:
\begin{align}\label{eqn: simple identity}
	g(m(\pi)) = g(\bE[H]).
\end{align}

Formula \eqref{eqn: simple identity} is mathematically straightforward, but the right-hand side of \eqref{eqn: simple identity} is computationally more tractable than the left-hand side. To be more precise, one main difficulty in estimating $g(m(\pi))$ arises from the difficulty in sampling $\pi$. However, our observation is  the quantity $g(m(\pi))$ essentially depends only $m(\pi)$ -- an expectation under $\pi$, but not  $\pi$ itself. Therefore, the quantity $m(\pi)$ can be replaced by the expectation of any unbiased estimator of $m(\pi)$. In other words,
we can relax the previous assumption `i.i.d. samples from $\pi$' by `i.i.d. unbiased estimators of $m(\pi)$'. Suppose $H_1, H_2, \ldots$ are  i.i.d. unbiased estimators of $m(\pi)$ that we can sample from. 
Then it suffices to estimate $g(\bE[H_1])$ unbiasedly. The difficulty is now reduced to estimating a function of expectation, and the existing unbiased MLMC methods can be applied.

After observing \eqref{eqn: simple identity}, it suffices to construct unbiased estimators of $m(\pi)$ provided that $\pi$ cannot be directly simulated. The unbiased MCMC framework provides us with natural solutions. Suppose a Markov chain with transition kernel $P$ that targets $\pi$ as stationary distribution. It is often possible to construct a pair of coupled Markov chains  $(Y,Z) = (Y_t, Z_t)_{t=1}^\infty$ that both evolve according to $P$. By design, if the pair $(Y_t, Z_{t-1})$ meets at some random time $\tau$ and stays together after meeting, then the Jacob-O'Leary-Atchad{\'e} (JOA) estimator, which will be formally introduced in shortly later, is unbiased for $m(\pi)$. 
Putting the unbiased MLMC and JOA estimator together, we can unbiasedly estimate $g(m(\pi))$ using the following two-step strategy described in  Figure \ref{fig:workflow} below. The unbiased MCMC algorithm is used here as a generator for random variables with expectation $m(\pi)$. We will use the outputs of the unbiased MCMC algorithm as inputs to feed into the unbiased MLMC approach and eventually construct an unbiased estimator of $g(m(\pi))$.

\begin{figure}[h]
	\centering    
	\includegraphics[width=0.9\textwidth]{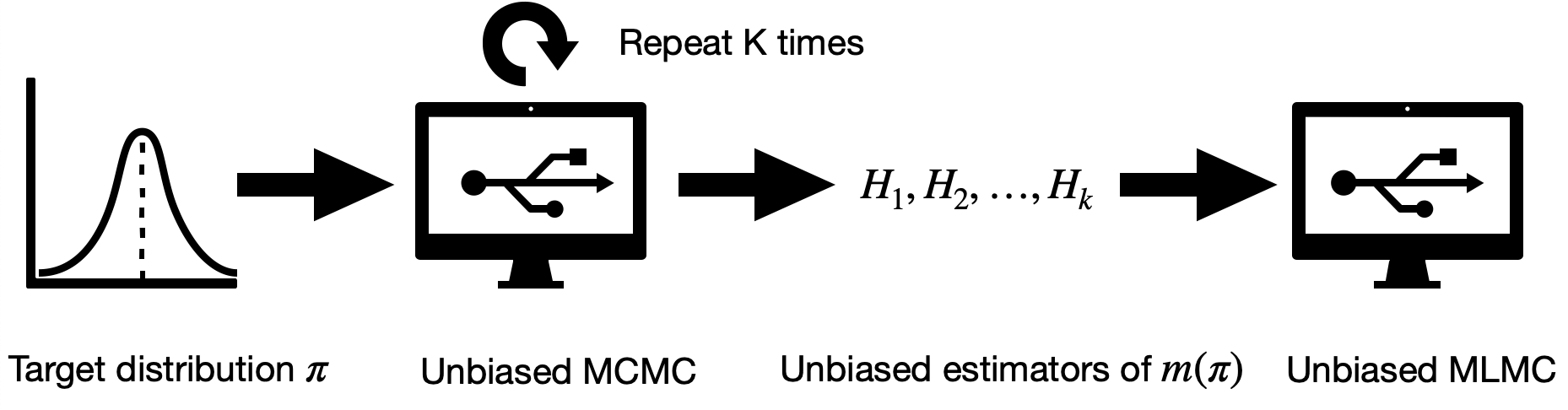} 
	\caption{The workflow for constructing an unbiased estimator of $g(m(\pi))$.}\label{fig:workflow}
\end{figure}

\section{Unbiased estimators for functions of expectation}\label{sec:estimator}
In this section, we discuss our estimator for $g(m(\pi))$ from MCMC outputs in detail. 
We start with a brief review of the JOA estimator of $m(\pi)$ in Section \ref{subsubsec:unbiased MCMC}. 
Our general framework is described in Section \ref{subsec: general g}. A family of simplified estimators is given in Section \ref{subsec: special g}  when $g$ admits additional structures. In Section \ref{subsec:nest}, we discuss the unbiased estimation of nest expectations using a generalized version of our approach. In Section \ref{subsec:the domain problem}, we discuss the problem regarding the domain of $g$ and provide a transformation to avoid the domain problem. In Section \ref{subsec:theory}, we give theoretical justifications for our method. 
\subsection{Constructing  the unbiased estimator}\label{subsec:construction}
\subsubsection{The Jacob-O'Leary-Atchad{\'e} (JOA) estimator of $m(\pi)$} \label{subsubsec:unbiased MCMC}
Let $\Omega$ be a Polish space equipped with the standard Borel $\sigma$-algebra $\calF$. Let $P: \Omega\times \calF \rightarrow [0,1]$ be the Markov transition kernel that leaves $\pi$ as stationary distribution. The Jacob-O'Leary-Atchad{\'e} (JOA) estimator uses a coupled pair of Markov chains that both have transition kernel $P$. Formally, the coupled pair $(Y,Z) = (Y_t, Z_t)_{t=1}^\infty$ is a Markov chain on the product space $\Omega\times \Omega$. The transition kernel $\bar P$, which is also called the coupling of $(Y,Z)$,  satisfies $\barP((x,y), A\times \Omega) = P(x,A), \barP((x,y), \Omega\times B) = P(y,B)$ for every $x,y \in \Omega$ and $A,B\in \calF$. The coupled chain starts with $Y_0\sim \pi_0, Y_1\sim P(Y_0,\cdot )$ and $Z_0\sim \pi_0$ independently. Then at each step $t\geq 2$, one samples $(Y_t, Z_{t-1}) \sim \barP((Y_{t-1}, Z_{t-2}), \cdot)$. Suppose the coupling $\bar P$ is `faithful' \citep{rosenthal1997faithful}, meaning that there is a random but  finite time $\tau$ such that $Y_{\tau} = Z_{\tau-1}$, and $Y_t = Z_{t-1}$ for every $t\geq \tau$. Then for every $k$, the estimator
$	H_k(Y,Z):= f(Y_k) + \sum_{i=k+1}^{\tau-1}(f(Y_i) - f(Z_{i-1}))$ is  unbiased for $\bE_\pi[f]$.   The following informal calculation shows the unbiasedness in \cite{jacob2020unbiased}:
\begin{align*}
	m(\pi) = \lim_{n\rightarrow \infty} \bE[f(Y_n)] &=  \bE[f(Y_k)] + \sum_{n=k+1}^\infty (\bE[f(Y_{n})] - \bE[f(Y_{n-1})])\\
	&= \bE[f(Y_k)] + \sum_{n=k+1}^\infty\bE[f(Z_{n}) - f(Y_{n-1})] \\
	&= \bE[f(Y_k)] + \sum_{n=k+1}^{\tau-1}\bE[f(Z_{n}) - f(Y_{n-1})] = \bE[H_k(Y,Z)].
\end{align*}
The rigorous proof requires assumptions on the target $\pi$ and the distribution of $\tau$, see \cite{jacob2020unbiased, middleton2020unbiased} and our appendix for details. In principle, the above construction works for arbitrary initialization $\pi_0$, though the efficiency depends crucially on the initialization. In practice, users typically choose $\pi_0$ in the same way as they initialize their standard MCMC algorithm. Furthermore, for any fixed integer $m \geq k$, the  `time-averaged' estimator $H_{k:m}(Y,Z): = (m-k+1)^{-1}\sum_{l=k}^m H_l(Y,Z)$ clearly retains unbiasedness and reduces the variance. In practice, users typically choose $k$ to be a large quantile of the coupling time and $m$ to be several multiples of $k$. Theoretical and empirical investigations of these methods are provided in \cite{o2021metropolis, wang2021maximal}. More sophisticated estimators using $L$-lag coupled chains are discussed in  \cite{biswas2019estimating}, but the main idea remains the same.

\subsubsection{Unbiased estimator of $g(m(\pi))$}\label{subsec: general g}

Suppose we can access a routine $\cal S$ such as the JOA estimator in Section \ref{subsubsec:unbiased MCMC}, which outputs unbiased estimators of $m(\pi)$. The estimator of $g(m(\pi))$ can then be constructed by the randomized MLMC method. Let $H_1, H_2, \cdots, H_{2m}$ be a sequence of $i.i.d.$ random variables. We let $	S_H(2m) := \sum_{k=1}^{2m}H_i$ be the summation of all the $2m$ terms, and let
$	S_H^{\sfO}(m) :=\sum_{k=1}^m H_{2k-1}, S_{H}^{\sfE}(m) := \sum_{k=1}^m H_{2k}$
be the summation of  all the odd and even terms, respectively. Our estimator is described by Algorithm \ref{alg:MLMC}.

\begin{algorithm}[htbp]
	\caption{Unbiased Multilevel Monte-Carlo estimator}\label{alg:MLMC}
	
	\begin{algorithmic}
		
		\State \textbf{Input:} 
		\begin{itemize}
			\item A subroutine $\calS$ for generating unbiased estimators of $m(\pi)$
			\item A function $g: \calD \rightarrow \bR$
			\item  The parameter $p$ for geometric distribution
		\end{itemize}
		
		\State \textbf{Output:} Unbiased estimator of $g(m(\pi))$ \\
		\begin{enumerate}
			\item Sample $N$ from the geometric distribution $\Geo(p)$
			\item Call $\calS$ for $2^{N}$ times and label the outputs by $H_1, ..., H_{2^{N}}$ 
			\item Calculate the  quantities $S_H(2^N)$, $S_H^{\sfO}(2^{N-1})$ and $S_{H}^{\sfE}(2^{N-1})$ defined above
			\item Calculate $\Delta_N = g\left(S_H(2^{N})/2^N\right) -\frac{1}{2}\left(g\left(S_H^\sfO(2^{N-1})/2^{N-1}\right) + g\left(S_H^\sfE(2^{N-1})/2^{N-1}\right)\right)$
		\end{enumerate}
		\State\textbf{Return: } $W=\Delta_N/p_N + g(H_1)$.
	\end{algorithmic}
\end{algorithm}

Now we discuss the construction of our estimator $W$. Our approach is closely related to the Blanchet--Glynn estimator \citep{blanchet2015unbiased}. The critical difference is that our method relaxes the assumption `$i.i.d.$ samples from $\pi$' by `unbiased estimator of $m(\pi)$' and incorporates the JOA estimator as a subroutine. Since exact sampling from $ \pi $ is generally challenging, this relaxation is crucial for practical applications.

After rewriting $g(m(\pi)) = g(\bE[H])$, the core idea is to write $g(\bE[H])$ as the limit of a sequence of expectations. Here we use the Law of Large Numbers (LLN) and write $$g(\bE[H]) =  \bE[\lim_{n\rightarrow\infty} g(S_H(2^n)/2^n)] = \lim_{n\rightarrow\infty} \bE[g(S_H(2^n)/2^n)].$$
After introducing our technical assumptions, we will justify the validity of exchanging the order between the expectation and limit. Then one can write the limit of expectations as an infinite summation of consecutive sums, i.e., 
\begin{align*}
   g(E[H]) =  \lim_{n\rightarrow\infty} \bE[g(S_H(2^n)/2^n)] &= \bE[g(H_1)] + \sum_{n=1}^\infty \bE[g(S_H(2^n)/2^n)] - \bE[g(S_H(2^{n-1})/2^{n-1})]\\
    &= \bE[g(H_1)]  +  \sum_{n=1}^\infty  \bE[\Delta_n],
\end{align*}
 where $\Delta_n$ is defined in Step 4 in Algorithm \ref{alg:MLMC}. For each fixed $n$, the random variable $\Delta_n$ can be simulated with cost $2^n$. To tackle the infinite summation of the expectations, one can  choose a random level $N$ with probability $p_N$ and construct the importance sampling-type estimator $\Delta_N/p_N$. The following informal calculation justifies the unbiasedness of  $W$ (output of Algorithm \ref{alg:MLMC}). 
\begin{align*}
	\bE\left[W\right] &= \bE\left[g(H_1)\right] + \bE\left[\Delta_N/p_N\right] =  \bE\left[g(H_1)\right] +\bE [\bE\left[\Delta_N/p_N\mid N\right]]\\
	& = \bE\left[g(H_1)\right] + \sum_{n=1}^\infty \bE[\Delta_n] = \bE\left[g(H_1)\right] + \sum_{n=1}^\infty (\bE[g(S_H(2^n)/2^n)] - \bE[g(S_H(2^{n-1})/2^{n-1})])\\
	& = \lim_{n\rightarrow\infty} \bE[g(S_H(2^n)/2^n)] = g(\bE[H]) = g(m(\pi)).
\end{align*}

Moreover, constructing $\Delta_n$ is a crucial step in Algorithm \ref{alg:MLMC}. The construction in Step 4 of Algorithm \ref{alg:MLMC}
is often referred to as the `antithetic difference estimator,' which is also used in \cite{giles2014antithetic, blanchet2015unbiased}.  A natural question is whether one can replace the antithetic difference design with the following seemingly  more straightforward  estimator:
$
\tilde \Delta_n = g\left(S_H(2^{n})/2^n\right) - g\left(S_H(2^{n-1})/2^{n-1}\right).
$
It turns out we cannot. The rationale behind the antithetic difference design is that we want to control both the variance and computational cost simultaneously. As we will see from Section \ref{subsec:theory}, the antithetic difference design allows one to cancel both the constant and linear terms in the Taylor expansion.   In contrast, $\tilde \Delta_n$ only cancels the constant term. This difference eventually implies our unbiased estimator (output of  Algorithm \ref{alg:MLMC})  will have both finite variance and finite computational cost only if we use the antithetic difference design.

It may seem daunting that  Algorithm \ref{alg:MLMC} generates $2^N$ samples for each implementation. However, the actual computational cost is reasonable as the random variable $N$ follows a geometric distribution and therefore has an exponentially light tail that compensates for the exponentially increasing term $2^N$. To be more precise, suppose it takes unit cost to call $\calS$ once,  in several practical cases including \cite{Blanchet2015UnbiasedMC}, the authors choose $p = 1 - 2^{-1.5}\approx 0.646$, the expected computational cost for implementing Algorithm \ref{alg:MLMC} once is then around $
\sum_{n = 0}^\infty 2^n (1-p)^{n-1} p = \frac{p}{1-p}\sqrt{2}\approx 2.580.$
Therefore, the expected cost of Algorithm \ref{alg:MLMC} is shorter than calling the subroutine $\calS$ three times. Detailed discussion on the computational cost and the choice of $p$ can be found in Section \ref{subsec:theory}.

We use the JOA estimator in Algorithm \ref{alg:MLMC} as our algorithm needs a subroutine to sample unbiased estimators of $m(\pi)$. In principle, any unbiased estimator of $m(\pi)$ (see, e.g., \cite{agapiou2018unbiased,ruzayqat2022unbiased}) can also be fed into Algorithm \ref{alg:MLMC} as a subroutine. On the other hand,  the JOA estimator is by far the most general framework for constructing unbiased estimators of $m(\pi)$ given intractable $\pi$. For concreteness, we will assume the subroutine $\calS$ is the JOA estimator subsequently.

\subsubsection{Unbiased estimator of  polynomials and other special functions}\label{subsec: special g}
Section \ref{subsec: general g} provides us a relatively general framework for unbiased estimators of $g(m(\pi))$. In some situations where the target function $g$ has certain nice properties, the unbiased estimators can be easily obtained without resorting to the unbiased MLMC framework. For example, if $g(x) = x^k$ is a univariate monomial function, one can call the unbiased MCMC algorithm $k$ times and obtain unbiased estimators $H_1, \cdots, H_k$ of $\bE_\pi[X]$. The estimator $\prod_{l=1}^k H_l$ will then be unbiased for $m(\pi)^k$. The argument above can be naturally extended to the case where $m(\pi) \in \bR^m$ and $g: \bR^m \rightarrow \bR$ is a multivariate polynomial function. We use the multi-index $k = (k_1,\cdots, k_m)$ with $\sum_{i=1}^m k_i\leq n$ where $k_1, \ldots, k_m$ are non-negative integers, and $x^k = x_1^{k_1} x_2^{k_2} \cdots x_m^{k_m}$. Let $g(x) = \sum_{k\leq n}  \alpha_k x^k$ denote a multivariate polynomial with degree at most $n$. The unbiased estimator of $g(m(\pi))$ can be constructed as follows. First, we call the unbiased MCMC subroutine $\calS$ for $n$ times and label the outputs by $H_1, \cdots, H_n$, each is an independent vector-valued unbiased estimator of $m(\pi)$. Then for each $k = (k_1, \cdots, k_m)$ we calculate the quantity 
$\hat H(k) =\prod_{l_1 = 1}^{k_1}H_{l_1,1} \prod_{l_2 = k_1+1}^{k_1+k_2}H_{l_2,2} \cdots \prod_{l_m = k_1 + \cdots + k_{m-1}+ 1}^{k_1 + \cdots + k_m} H_{l_m, m},$
where $H_{a,b}$ stands for the $b$-th coordinate of $H_a\in \bR^d$. It is  clear from the independence of $H_1, \cdots, H_n$ that $\bE[\hat H(k)] = m(\pi)^k$. Finally, we output $\sum_k  \alpha_k \hat H(k)$, which is unbiased for $g(m(\pi))$ by the linearity of expectation. It is different from Algorithm \ref{alg:MLMC} as it requires a fixed number of calls for $\cal S$.

When $g: \bR \rightarrow \bR$ is a real analytic function on $\cal D$, i.e., $g(x) = \sum_{n=0}^\infty a_i (x - a)^n$ for some real number $a$. Suppose $\tilde N$ is a non-negative integer random variable with $\bP(\tilde N = k) = q_k$. The unbiased estimator for $g(m(\pi))$ can be constructed by first generating $\tilde N$, and then calling the subroutine $\calS$ for $\tilde N$ times to generate unbiased estimators of $\bE_\pi[X]$. Denote the outputs by $H_1, \cdots H_{\tilde N}$, the final  estimator can be expressed by $(a_{\tilde N}/q_{\tilde N}) \cdot (\prod_{j=1}^{\tilde N}(H_j - a)/\tilde N!)$
This idea exists in previous literature, such as \cite{blanchet2015unbiased}, when $\pi$ can be perfectly simulated. We generalize this idea to the case where $\pi$ is intractable. In particular, when $g(x) = e^x$ and $\tilde N$ follow from the Poisson distribution, the estimator is known as the `Poisson estimator,' which is used in both physics and statistics, see \cite{wagner1987unbiased,papaspiliopoulos2009methodological, fearnhead2010random}. 

Albeit useful in many cases, the power-series-type estimators generally have strong assumptions about the smoothness of the target function. It also requires the knowledge of all the higher-order derivatives of $g$, which is generally infeasible when $g$ is complicated. Therefore, throughout this paper, we will primarily focus on using the unbiased MLMC framework for estimating $g(m(\pi))$ given its generality. This subsection intends to remind our readers that more straightforward choices may exist when  $g$ behaves `nice' enough.  

\subsection{Nested Expectations}\label{subsec:nest}
Now we  extend our method to estimate the nested expectations. Recall that a nested expectation can be written as  $\bE_\pi[\lambda]$, where $\lambda(x) := f(x, \gamma(x))$, where $\gamma(x) = \bE_{y\sim \pi(y\mid x)}[\phi(x,y)]$ is another expectation under the conditional distribution. We first decompose the joint distribution $\pi(x,y)$ as the marginal distribution $\pi(x)$ times the conditional distribution of $\pi(y|x)$. When fixing $x = x_0$, then $\lambda(x_0) = f(x_0, \bE_{y\sim \pi(y\mid x_0)}[\phi(x_0,y)])$ is  a function of $\bE_{\pi(y|x_0)}[\phi(x_0,y)]$ and our previous framework can be applied. Our   estimator is as follows.

\begin{algorithm}[htbp]
	\caption{Unbiased Multilevel Monte-Carlo estimator for nested expectation}\label{alg:MLMCnest}
	
	\begin{algorithmic}
		\State 
		\begin{enumerate}
			\item Sample $x$ from $\pi(x)$
			\item Given $x$ fixed, generate an unbiased estimator   $\hat \lambda(x)$ of $\lambda(x)$ using Algorithm \ref{alg:MLMC}
		\end{enumerate}
		\State\textbf{Return: } $\hat \lambda(x)$.
	\end{algorithmic}
\end{algorithm}
 Algorithm \ref{alg:MLMCnest} can be viewed as the `conditional' version of Algorithm \ref{alg:MLMC}. We first sample $x$ and apply Algorithm \ref{alg:MLMC} to generate an unbiased estimator under  $\pi(\cdot | x)$. After taking the randomness of $x$ into account, we show the output Algorithm  \ref{alg:MLMCnest}  is unbiased for $\bE_\pi[\lambda]$.
\begin{prop}\label{prop:nested unbias}
	We have $\bE[\hat \lambda] = \bE_\pi[\lambda]$.
\end{prop}
The proof of Proposition \ref{prop:nested unbias} is given in Appendix \ref{subsubsec: proof nested}.

Algorithm \ref{alg:MLMCnest} is useful when $\pi(x)$ can be directly sampled from, and $\pi(y | x)$ can be approximated sampled from some MCMC algorithms. To see the potential applications of Algorithm \ref{alg:MLMCnest},  we present a typical example of the nested expectation, namely estimating the expected utility under partial information  \citep{giles2018mlmc, giles2019decision}. Other examples, including the Bayesian experimental design and variational autoencoders, are given in \cite{rainforth2018nesting, hironaka2021efficient, goda2022unbiased}.

\begin{example}[The  utility under partial information]
	Suppose we have a two-stage process $(X,Y)$ with joint distribution $\pi(x,y)$. Suppose we have $D$ possible strategies (for example, treatments), each with corresponding utility $f_d(x,y)$ for $d \in\{1, \cdots, D\}$. If we have to choose a strategy without seeing the values of $(X,Y)$, the optimal expected utility would be $\max_d \bE[f_d(X,Y)]$. Similarly, after seeing the whole information, the optimal utility would be $\bE[\max_d f_d (X, Y)]$. In the intermediate case, if one has observed only $X$,  the optimal strategy would maximize the conditional utility, i.e., $d^*(X) = \arg\max_d \bE[f_d(X, Y)|X]$. The optimal utility with partial information is $\bE[\max_d \bE[f_d(X, Y)|X]]$, which is a nested expectation. 
	
	The expected utility under partial information finds applications in computational finance, especially in option pricing \cite{belomestny2015multilevel, zhou2021unbiased}. Meanwhile, the  difference between full and partial utility 
	$ \bE[\max_d f_d (X,Y)] - \bE[\max_d \bE[f_d(X,Y)|X]]$ quantifies the `value' of the information in $Y$, which also has applications in the evaluation of Value-at-Risk (VaR) \citep{giles2018mlmc} and medical areas \citep{ades2004expected}. Existing literature typically assumes one can sample directly from  $\pi(x,y)$, and regard the intractable $\pi(x,y)$ as an open question, see Section 5 of \cite{giles2019decision} for discussions.
\end{example}

\subsection{The domain problem and the $\delta$-transformation}\label{subsec:the domain problem}

There is an extra subtly in implementing Algorithm \ref{alg:MLMC}. Besides requiring $H$ to be an unbiased estimator of $m(\pi)$, Algorithm \ref{alg:MLMC} implicitly requires the range of $S_H(m)/m$ to be a subset of the domain of $g$. This constraint is naturally satisfied when $g: \calD \rightarrow \bR$ has domain $\calD = \bR^m$, such as $g(x) = e^x$, or $g(x_1,x_2) = \max\{x_1,x_2, 1\}$. However, many natural functions are not defined on the whole space, such as $g(x) = 1/x$, or $g(x_1,x_2) = x_1/ x_2$. These functions arise in statistical applications such as doubly-intractable problems \citep{lyne2015russian}, estimating the ratio of normalizing constants \citep{meng1996simulating}. Unfortunately, Algorithm \ref{alg:MLMC} cannot be implemented if  $S_H(m)/m$ falls outside $\calD$.  

Consider a concrete problem of estimating $g(m(\pi))) = 1/m(\pi)$ where $\pi$ is a probability measure on $\Omega$. The problem can be naturally avoided if  $S_H(m)/m \neq 0$ almost surely, which is often the case for continuous state-space $\Omega$. However, the algorithm may fail for discrete state spaces. Even if $\Omega$ only contains positive numbers, the resulting JOA estimator may still take $0$ with positive probability. The same problem  gets worse if the domain of $g$ is of the form $\{x\mid\lVert x\rVert \geq c\}$, where both continuous and discrete Markov chains may fail. 

We add an extra $\delta$-transformation to address this issue when needed. Suppose $\calD \supset \bR^d \setminus B_{\delta}$, where $B_{\delta}:= \{x\mid \lVert x \rVert \leq \delta\}$. In other words, $\calD$ contains everything in $\bR^d$ except for a compact set. Let $H$ be the output of the unbiased MCMC subroutine $\calS$. If $\lVert H \rVert \leq \delta$, we flip a fair coin and move $H$ to $H + 2\delta$ given head and $H- 2\delta$ given tail. Formally  the transformation can be defined as $H \rightarrow\tilde H := H 1_{\lVert H\rVert\geq \delta} +  (H + 2\delta B)1_{\lVert H\rVert < \delta}$, where $B$ follows a uniform two-point distribution on $\{-1,1\}$.  After the transformation, $\tilde H$ has support in $\calD$, and the next proposition shows $\tilde H$ has the same expectation as $H$ (and therefore still unbiased), with variance no larger than the variance of $H$ plus an absolute constant.

\begin{prop}\label{prop:delta transform} Let $\tilde H$ be $\delta$-transformation of $H$, then $\lVert\tilde H\rVert \geq \delta$ and $\bE[\tilde H] = \bE[H] = m(\pi)$, and $\var[\tilde H] =  \var[H] + 4\delta^2\bP[\lVert H\rVert\geq \delta]\leq \var[H] + 4\delta^2 $.
\end{prop}

The $\delta$-transformation  can be used after Step 2 of Algorithm \ref{alg: MLMC}  for the outputs of the unbiased MCMC algorithm. After getting  $H_1, \ldots, H_{2^N}$ of $m(\pi)$, we could apply the $\delta$-transformation on each of them to ensure every $\tilde H_i$ is still unbiased but has support inside $\calD$. Since the above proposition shows the $\delta$-transformation only increases the variance by no more than $4\delta^2$, theoretical results in Section \ref{subsec:theory} below also hold for estimators after the transformation, albeit a slightly worse dependency on the constants.

\subsection{Theoretical results}\label{subsec:theory}

With all the notations above, we are ready to state our technical assumptions and prove the theoretical results. Our theoretical analysis will focus on the unbiased estimator described in Algorithm \ref{alg:MLMC}. All the results still go through if the $\delta$-transformation is needed. Recall that $g$ is a function from $\calD$ to $\bR$, and $H_1, H_2, \cdots $ are $i.i.d.$ unbiased estimators of $m(\pi)$. Now we denote by $V_n\subset \mathbb R^d$ the range of $(H_1 + \cdots + H_n)/n$ for every $n$ and $V := \cup_{n=1}^\infty V_n$. Our assumptions are posed on both $g$ and $H_i$:

\begin{assumption}[Domain]\label{ass: domain}
	The function $g: \calD\rightarrow \mathbb R$ satisfies $V\subset \calD$. Moreover,  $m(\pi)$ is in the interior of $\calD$, i.e., $m(\pi) \in \calD^\circ$.
\end{assumption}

\begin{assumption}[Consistency]\label{ass: consistency}
	$\bE[g(S_H(n)/n)]\rightarrow g(m(\pi))$ as $n\rightarrow \infty$.
\end{assumption}

\begin{assumption}[Smoothness]\label{ass: differentiable}
	The function $g$ is continuously differentiable in a neighborhood of $m(\pi)$, and $Dg\left(\cdot\right)$ is locally  H\"{o}lder continuous with exponent $\alpha>0$. In other words, there exists $\varepsilon > 0$, $\alpha > 0$ and $c = c(\epsilon) >  0$ such that s for every $x,y \in (m(\pi) - \epsilon, m(\pi) + \epsilon$),
	$\lVert Dg(x)-Dg(y)\rVert\leq c\lVert x-y\rVert^{\alpha}.$ 
\end{assumption}

\begin{assumption}[Moment]\label{ass: moment}
	There exists some $l > 2+\alpha$ such that $H$ has finite $l$-th moments, i.e., $\mathbb{E}[\lVert H_1\rVert^{l}_l] = \sum_{i=1}^m \bE[\lvert H_{1,i}\rvert^l]<\infty.$
\end{assumption}

\begin{assumption}[Smoothness--Moment Tradeoff]\label{ass:tradeoff}
	There exist constants $s > 1$, $\alpha_s  \in \bR$, and $\calC_s > 0$ such that $2\alpha_s + (s-1)l > 2s$ and $\mathbb E(\lvert \Delta_n\rvert^{2s}) \leq \calC_s 2^{-\alpha_s n}$ for every $n\geq 0$, where 
	$$\Delta_n =\begin{cases}
		g\left(S_H(2^{n})/2^n\right) -\frac{1}{2}\left(g\left(S_H^{\sfO}(2^{n-1})/2^{n-1}\right) + g\left(S_H^{\sfE}(2^{n-1})/2^{n-1}\right)\right) \qquad & n\geq 1\\
		g(H_1) \qquad &n = 0.
	\end{cases} $$
\end{assumption}

We briefly comment on the Assumptions \ref{ass: domain} -- \ref{ass:tradeoff}. The descriptions below are mostly pedagogical, and the detailed proofs are deferred to the Appendix (Section \ref{sec:appendix}).

The Domain Assumption \ref{ass: domain} guarantees Algorithm \ref{alg:MLMC} can be implemented. When $g$ does not directly satisfy this assumption, but $\calD \supset \bR^d \setminus B_{\delta}$, then we apply the $\delta$-transformation to enforce the first half of Assumption \ref{ass: domain} holds. All the theoretical results still hold.

The consistency Assumption \ref{ass: consistency} is expected and somewhat necessary. It appears in related works, including \cite{vihola2018unbiased, Blanchet2015UnbiasedMC} explicitly or implicitly. The Law of Large Numbers guarantees $S_H(n)/n \rightarrow m(\pi)$, therefore $g(S_H(n)/n) \rightarrow g(m(\pi))$ due to the continuity. Assumption \ref{ass: consistency} is generally satisfied by the dominated convergence theorem. 

The Smoothness Assumption \ref{ass: differentiable} guarantees both $g$ is smooth enough at a neighborhood of $m(\pi)$, and the derivative of $g$ is  H\"{o}lder continuous. When $g$ is infinitely differentiable, and there is no singularity on a neighborhood of $m(\pi)$, then we expect Assumption \ref{ass: differentiable} to hold with $\alpha \geq 1$.  We emphasize that we only require $Dg$ to be locally H\"{o}lder continuous near $m(\pi)$, which is much weaker than requiring $Dg$ to be globally H\"{o}lder continuous. 

The  Moment Assumption \ref{ass: moment} requires  more than $l$-th moment of the unbiased estimator $H_i$, where $l$ is strictly larger than $2+\alpha$.  When the JOA estimator is used for generating $H_i$, Assumption \ref{ass: moment} generally holds when $f$ has strictly more than $l$-th moment under $\pi$, and the coupling time $\tau$ has a very light tail. The tail behavior of $\tau$ is closely related to the mixing time of the underlying MCMC algorithm. We recall that a $\pi$-stationary Markov chain with transition kernel $P$ is said to be geometrically ergodic if there is a $\gamma\in(0,1)$ and a function $C: \Omega \rightarrow (0,\infty)$ such that $\lVert P^n(x,\cdot) - \pi \rVert_\TV \leq C(x) \gamma^n,$
for $\pi$--a.s. $x$. Geometric ergodicity is a central notion in MCMC theory. There is a large body of literature, including but not limited to, \cite{mengersen1996rates, roberts1996exponential, roberts1996geometric, wang2020theoretical, livingstone2019geometric}, that shows 
a wide family of MCMC algorithms is geometrically ergodic. 

Our result for guaranteeing Assumption \ref{ass: moment} is  the following.
\begin{prop}[Verifying Assumption \ref{ass: moment}, informal]\label{prop: moment, informal}
	Suppose the Markov chain $P$ is $\pi$-stationary and geometrically ergodic, and $f$ is a measurable function with finite $p$-th moment under $\pi$ for any $p>l$. Suppose also  there exists a set $\calS\subset \Omega$, a constant $\tilde \epsilon \in (0,1)$ such that $
	\inf_{(x,y)\in \calS \times \calS} \barP((x,y), \calD) \geq \tilde\epsilon,$
	where $\calD:= \{(x,x): x\in \Omega\}$ is the diagonal of $\Omega\times \Omega$. Then the JOA estimator  $H_k(Y,Z):= f(Y_k) + \sum_{i=k+1}^{\tau-1}(f(Y_i) - f(Z_{i-1}))$ has a finite $l$-th moment, and therefore satisfies Assumption \ref{ass: moment}. 
\end{prop}
The formal description of the above proposition and the detailed proofs will be deferred to Appendix \ref{subsec: moment assumption and mixing}. It can be viewed as a slightly stronger version of Proposition 3.1 in \cite{jacob2020unbiased}, where the authors established the finite second-order moment.

The Tradeoff Assumption \ref{ass:tradeoff} bounds $\mathbb E[\lVert \Delta_n\rVert^{2s}]$. The condition $2\alpha_s + (s-1)l > 2s$  reflects the tradeoff between the smoothness of $g$ and the moment assumption on $H_i$. Consider the following scenarios: \textit{1:}  Suppose $g$ is at least twice continuously differentiable, and the derivative $Dg$ is Lipschitz continuous. Then we have $\Delta_n = \calO((S_H(2^n)/2^n)^2)$ by Taylor expansion. Meanwhile, the Central Limit Theorem (CLT) shows $\Delta_n = \calO_p(2^{-n})$. Therefore we choose $\alpha_s = s$, and  Assumption \ref{ass:tradeoff} is true for positive $l$. In this case, Assumption \ref{ass:tradeoff} is weaker than  \ref{ass: moment}. \textit{2:} Suppose $g$ is at most of  linear growth, i.e., $\lvert g(x) \rvert \leq c(1 + \lVert x\rVert)$. In this case we can only bound $\Delta_n$ by $\calO_p(2^{-n/2})$ again by the CLT.  We  choose $\alpha_s = s/2$  and it thus requires $l > s/(s-1)$. This is also the assumption in \cite{Blanchet2015UnbiasedMC, blanchet2019unbiased}. \textit{3:} Suppose  $\mathbb E[\lVert \Delta_n\rVert^{2s}]$ is uniformly bounded. Then we expect to choose $\alpha_s = 0$, and therefore $l > 2s/(s-1)$. In summary, stronger smoothness requirements on $g$ result in weaker assumptions on the moment of $H_i$, and vice versa.

Our main theoretical result is as follows.
\begin{theorem}\label{thm:bound-estimator}
	Under Assumption \ref{ass: domain} -- \ref{ass:tradeoff}, let  $\gamma := \min\{\alpha, \frac{\alpha_s}{s} + \frac{(s-1)l}{2s} - 1\} > 0.$ 	if $N\in\{1,2,\ldots\}$  is  geometrically distributed with success parameter $p \in \left(\frac{1}{2},1-\frac{1}{2^{(1+\gamma)}}\right)$, then the estimator $W:= \frac{\Delta_N}{p_N}+ g(H_1)$  described in Algorithm \ref{alg:MLMC} satisfies:
	\begin{enumerate}
		\item $\bE[W] = g(m(\pi))$,
		\item There exists a constant $C$ such that 
		$\var(W) \leq \bE[W^2] \leq Cp^{-1}\frac{2^{-(1+\gamma)}}{1 - \big((1-p)2^{1+\gamma}\big)^{-1}} <\infty.$
		\item The expected computational cost of Algorithm \ref{alg:MLMC} is finite.
	\end{enumerate}
\end{theorem}

The proof of Theorem \ref{thm:bound-estimator} relies on the following key lemma to bound $\Delta_n$:

\begin{lemma}\label{lem: delta-square}
	We have $\bE[\lvert \Delta_n\rvert ^2] = C2^{-(1+\gamma)n},$ where $\gamma = \{\alpha, \frac{\alpha_s}{s} + \frac{(s-1)l}{2s} - 1\} > 0,$ 
	and	$C = C(m,l, \epsilon, s, \alpha)$ is a constant provided that Assumption \ref{ass: domain} -- \ref{ass:tradeoff} are satisfied.
\end{lemma} 

The proof is deferred to Appendix \ref{subsec: bounding delta}, but the main idea is to use the antithetic design to cancel the linear term in the Taylor expansion. This cancellation in turn gives us $\bE[\lvert \Delta_n\rvert ^2] = \calO(2^{-(1+\Omega(1))n}),$ which has an $\calO(2^{-(\Omega(1))n})$ gain over the canonical rate from the CLT.  With Lemma \ref{lem: delta-square} in hand, we are ready to show Theorem \ref{thm:bound-estimator}.
\begin{proof}[Proof of Theorem \ref{thm:bound-estimator}]
	We will first show Statement $1$ assuming Statement $2$ holds. Then we show both Statement $2$ and $3$  holds. 
	
	\textit{Proof of Statement $1$}: Suppose $W$ has a finite second moment, then the conditional distribution $\bE[W| N]$ is well defined (see Section 4.1 of \cite{durrett2019probability}). By the law of iterated expectation:
$\bE[W] = \bE\big[\bE[W\mid N]\big] = \bE[g(H_1)] + \bE\big[\frac{\bE[\Delta_n\mid N]}{p_N}\big] =  \bE[g(H_1)] + \bE\big[d_N/p_N\big],$
	where $d_n = \bE[g(S_H(2^n)/2^n)] - \bE[g(S_H(2^{n-1})/2^{n-1})]$. We can further calculate $\bE\big[d_N/p_N\big]$:
	$
	\bE\big[d_N/p_N] = \sum_{i=1}^\infty (d_i/p_i) p_i = \sum_{i=1}^\infty d_i.
	$
	Therefore 
	$\bE[W] = \lim_{n\rightarrow\infty} \bE[g(S_H(2^n)/2^n)] = g(m(\pi)),$ as desired. The last equality uses Assumption \ref{ass: consistency}.
	
	\textit{Proof of Statement $2$}: Since $
	\bE[W^2] \leq 2 \big(\bE[g(H_1)^2] + \bE\big[\Delta_N^2/p_N^2\big]\big)
	$,
	it suffices to show $ \bE\big[\Delta_N^2/p_N^2\big] <\infty$. We have
	$
	\bE\big[\Delta_N^2/p_N^2\big] = \sum_{n=1}^\infty \bE[\Delta_n^2] (1-p)^{-n+1} p^{-1}.
	$
	By Lemma \ref{lem: delta-square},
	\begin{align*}
		\bE\big[\frac{\Delta_N^2}{p_N^2}\big] \label{eqn: variance}\leq Cp^{-1}(1-p)\sum_{n=1}^\infty 2^{-(1+\gamma)n} (1-p)^{-n}& = Cp^{-1}(1-p)\sum_{n=1}^\infty \big((1-p)2^{1+\gamma}\big)^{-n}\\
		&= Cp^{-1}\frac{2^{-(1+\gamma)}}{1 - \big((1-p)2^{1+\gamma}\big)^{-1}} < \infty,
	\end{align*}
	where the last inequality follows from $(1-p) > 2^{-(\gamma+1)}$.
	
	\textit{Proof of Statement $3$}:  Let $C_H$ be the computation cost for implementing the unbiased MCMC subroutine $\calS$ once. It is shown in \cite{jacob2020unbiased} that $C_H < \infty$. The computation cost for implementing Algorithm \ref{alg:MLMC} essentially comes from $2^N$ calls of the subroutine $\calS$, where $N\sim \Geo(p)$. Therefore it suffices to show $2^N$ has a finite expectation. We calculate
	\begin{align*}
		\bE[2^N] = \sum_{n=1}^\infty 2^n p(n) =   \sum_{n=1}^\infty 2^n (1-p)^{n-1}p = \frac{2p}{2p-1} < \infty,
	\end{align*}
	where the last inequality follows from $p > 1/2$.
\end{proof}

Theorem \ref{thm:bound-estimator} immediately implies the following corollary on the computational cost, with proof given in Appendix \ref{subsec:other proof}. 
The computation cost $O(1/\epsilon^2)$ is shown to be rate-optimal \cite{heinrich1992lower,dagum2000optimal} for  Monte Carlo estimators. 

\begin{corollary}\label{cor:application_of_unbiasedness}
	Under Assumption \ref{ass: domain} -- \ref{ass:tradeoff}, for any $\varepsilon >0$, we can construct an estimator $\tilde W$ within expected computational cost $\calO(1/\epsilon^2)$, such that the mean square error between $\tilde W$ and the ground truth $g(m(\pi))$ is bounded by $\epsilon^2$, i.e.
	$\bE[(\tilde W - g(m(\pi)))^2] \leq \epsilon^2.$
\end{corollary}

Now we discuss the choice of the parameter $p$  when implementing Algorithm \ref{alg:MLMC} in practice. Theorem \ref{thm:bound-estimator} suggests every $p \in (1/2, 1 - 1/2^{1+\gamma})$ guarantees unbiasedness, finite variance, and finite computational cost. On the other hand, a larger value of $p$ yields a faster completion time but a larger variance for obtaining one estimator using Algorithm \ref{alg:MLMC}. The actual choice of $p$ depends on the user's objective and the number of available processors. Here we discuss two practical scenarios:
\begin{itemize}
    \item Suppose the user has sufficiently many processors and wants to minimize the completion time. The users should choose the parameter $p$ as larger as possible (but no larger than the theoretical limit $1 - 1/2^{1+\gamma}$) to fully utilize their parallel computation capacity. To be precise, for fixed $p \in (1 - 1/2^{1+\gamma})$ and error tolerance level $\epsilon > 0$, `sufficiently many'  means more than $\var_p(W)/\epsilon^2$ processors, where $\var_p(W)$ is the variance of the output of Algorithm \ref{alg:MLMC} with input parameter $p$. In practice, the quantity $\var_p(W)$ is usually unknown to the users as a-priori. Nevertheless, users can either use the upper bound in Theorem \ref{thm:bound-estimator} as a conservative estimate or run some pre-experiments to estimate $\var_p(W)$. 
    \item Suppose the user wants to minimize the total computational cost over all the processors (which is different from the completion time when multiple processors are available). Then the objective is to minimize the work-normalized variance $\tilde{\sigma}_p^2(W)$  defined in  \cite{glynn1992asymptotic}, which is the product of the computation cost and the variance of an individual estimator. Then it follows from the above calculation that the $\tilde{\sigma}^2(W)$ is upper bounded by a constant multiple of $
	\sum_{n=1}^\infty \big((1-p)2^{1+\gamma}\big)^{-n} \times \sum_{n=1}^\infty \big(2(1-p)\big)^{n}.$ By Cauchy-Schwarz inequality, this upper bound can be minimized by choosing $p = 1 - 2^{-1 - \frac \gamma 2}$. When $\gamma = 1$, $p$ can be chosen as $ 1 - 2^{-\frac32}\approx 0.646$, recovering the result in \cite{Blanchet2015UnbiasedMC}. 
\end{itemize}

Finally, we present two Central Limit Theorems (CLTs) of our estimator. These results directly follow the standard arguments from \cite{glynn1991analysis, Blanchet2015UnbiasedMC}. These results show our estimator has the  `square-root' convergence rate. The CLTs can also help establish confidence intervals.

\begin{itemize}
	\item When the number of estimators $W_1, W_2, \ldots, W_n,\ldots$ in Algorithm \ref{alg:MLMC} goes to infinity, we have $
	\big(\frac{\sum_{i=1}^n W_i}{\sqrt n} - g(m(\pi)) \big)\rightarrow \sfN(0, \var(W_1))$ as $n\rightarrow \infty. $
	\item Given a fixed  budget $b$, let $N(b)$ be the number of $i.i.d.$ estimators $W_1, W_2, \ldots, W_{N(b)}$ that can be generated by time $b$. Then we have $
	\sqrt{b}\cdot\big(\frac{\sum_{i=1}^{N(b)} W_i}{N(b)}- g(m(\pi))\big) \rightarrow \sfN(0 , \tilde{\sigma}^2(W))$ as $b\rightarrow \infty$,
	where $\tilde{\sigma}^2(W)$ is the work-normalized variance defined above.
\end{itemize}

\section{Numerical examples}\label{sec:numerical}
Now we investigate the empirical performance of the proposed method with several examples. We first implement the algorithm on a multivariate Beta distribution and then on a $2$-D Ising model with periodic boundaries. In both examples, we compare the performance of our estimator with the standard Monte Carlo estimator when multiple processors are available. Finally, we estimate the nested expectations using a small real-data example modeled by the cut-distribution. Additional numerical experiment for estimating the inverse of natural statistics of the Ising model is presented in Appendix \ref{sec: extra experiment}. Throughout this section, the standard Monte-Carlo (or MCMC/Metropolis--Hastings/Gibbs sampler) estimator for $g(\bE_\pi[f])$ stands for the `plug-in' estimator  $g(\sum_{i=l}^n f(X_i)/n)$, where $\{X_i\}$ follows some MCMC algorithm targeting at $\pi$ with a burn-in period $l$. Fix any quantity $\mu$ that users wish to estimate, we define the relative error of an estimator $X$ as $\sqrt{\bE[(X - \mu)^2]}/|\mu|$.

\subsection{Product of inverse expectations}\label{subsec:product beta}
We begin with a toy model with known ground truth. Let $X = (X_1,\cdots, X_K)$ be a random vector with  independent components  $X_i \sim \mathsf{Beta}(i,1)$.  We are interested in the product of the inverse expectation:
$g_K\left(\bE[X]\right) =  \prod_{i=1}^K1/\bE[X_i]$. Standard calculation shows
$g_K\left(\bE[X]\right)  = K+1.$ Meanwhile, $g_K$ cannot be expressed as an expectation, so existing methods fail to provide unbiased estimators.

We apply our method to this problem. We first test the sensitivity of Algorithm \ref{alg:MLMC} to the parameter $p$, the success probability of the geometric distribution. Setting $K = 8,$ and using the R package `unbiasedmcmc' in \cite{jacob2020unbiased} for estimating $\bE[X_i]$ \footnote{Here the Beta distribution can be perfectly sampled, and there is no need to use the JOA estimator  in practice. However, for illustrating our general framework, we still implement the JOA estimators for estimating $\bE[X_i]$ via couplings of MCMC algorithms.}, we generate  $5\times 10^{4}$ unbiased estimates of $g_K\left(\cdot\right)$ using Algorithm \ref{alg:MLMC}  with parameter $p$ ranging from $0.6$ to $0.8$, $k = 4\times 10^4$ and $m = 4k$. Figure \ref{fig:beta_differentp} reports the relative and standard errors for each $p$. The plot shows that the estimates are pretty accurate and vary little for different $p$. We set $p=0.7$ in the following experiments to ensure high accuracy and efficient computation. Then we let $K$ change from $1$ to $8$ and test the accuracy of our method. For each $K$, we implement Algorithm \ref{alg:MLMC} for $5\times 10^4$ times independently to generate unbiased estimators of $g_K$. Our point estimates and the corresponding standard errors are reported in Figure \ref{fig:beta_differentk}. It is clear that the point estimates are highly accurate and fit the ground truth almost perfectly. The standard error gets larger when $K$ increases, indicating a higher uncertainty under higher dimensionality.

\begin{figure}[h]
	\begin{subfigure}{.5\textwidth}
		\centering
		\includegraphics[width=.8\linewidth]{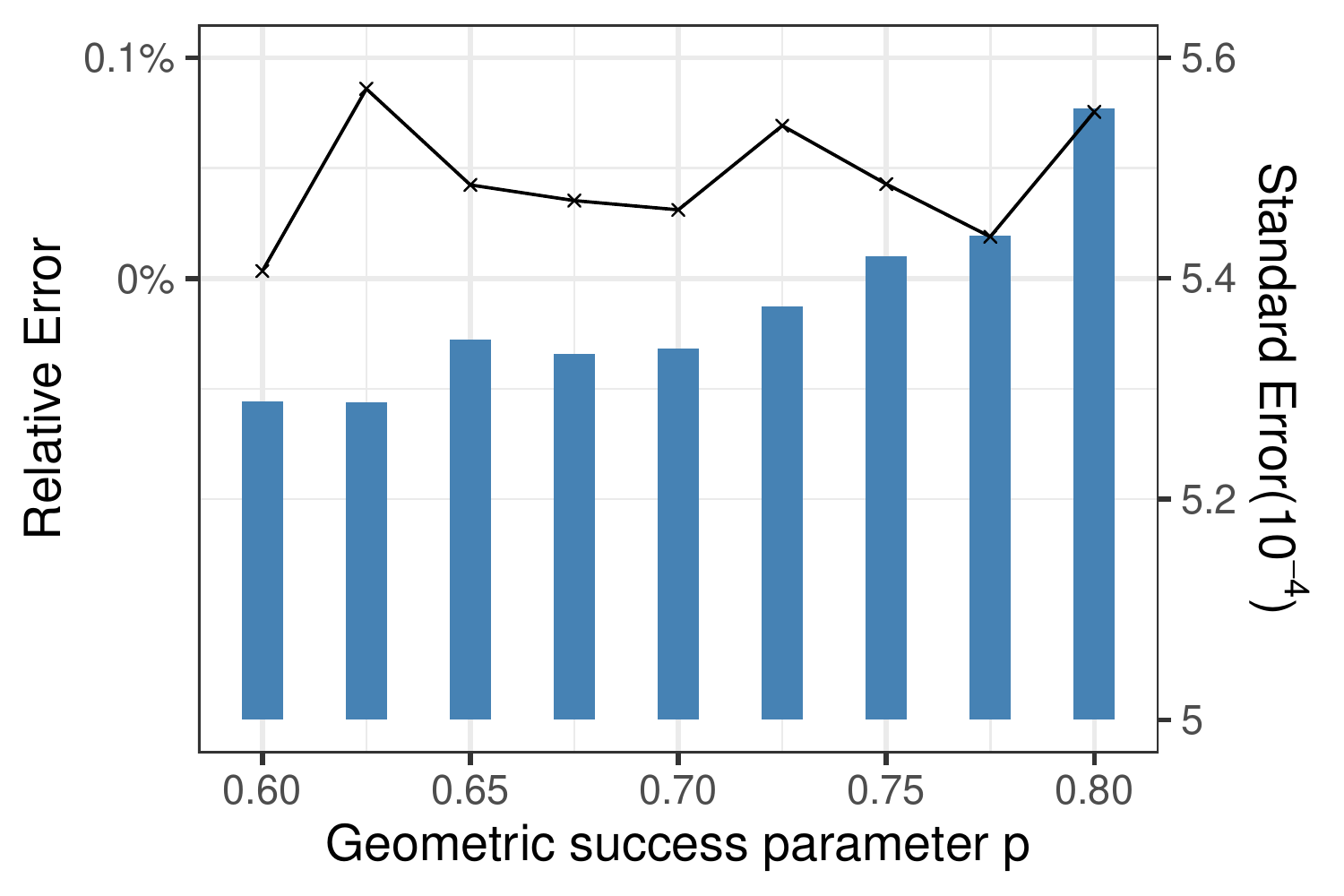}  
		\caption{$p$ changes, $K = 8$}
		\label{fig:beta_differentp}
	\end{subfigure}
	\hfill
	\begin{subfigure}{.5\textwidth}
		\centering
		\includegraphics[width=.8\linewidth]{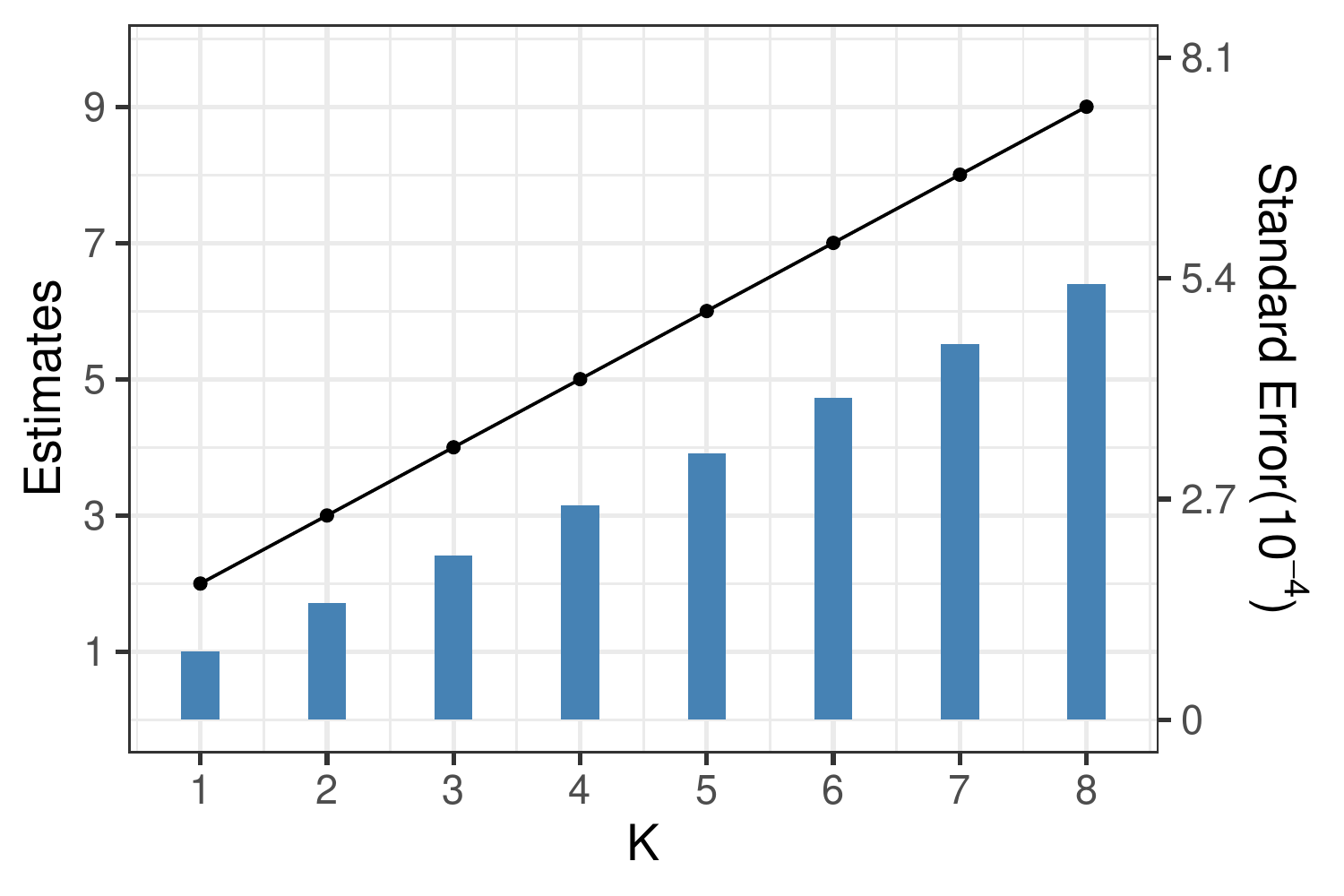}  
		\caption{$K$ changes, $p = 0.7$}
		\label{fig:beta_differentk}
	\end{subfigure}
	\caption{The relative error (line plot) and standard error (histogram) plots for $g_K$ based on $5\times 10^4$ unbiased estimators. Left: Fix dimension $K = 8$, parameter $p$ varies from $0.6 - 0.8$. Right: Fix parameter $p = 0.7$, dimension $K$ varies from $1$ to $8$.}
	\label{fig:multibeta}
\end{figure}

Now we compare our estimator with a Metropolis-Hastings estimator to show the performance of our method in the parallel regime. To make a fair computation, we use the same random-walk transition kernel in both the unbiased MCMC subroutine $\cal S$ of Algorithm \ref{alg:MLMC} and the MCMC algorithm. Since Algorithm \ref{alg:MLMC} takes a random computation time per run, we follow \cite{nguyen2022many} to ensure equal computation time across processors as follows: On each processor, we always first run  Algorithm \ref{alg:MLMC} and record its running time. Then we run the standard MCMC algorithm for the same time and discard the first $10\%$ samples as burn-in. This way, the two algorithms have the same computational cost for each processor. Finally, we run both methods independently on multiple processors and compare their accuracy after averaging their results respectively over all the processors.

Figure \ref{fig:beta_box} depicts the different bias/variance behaviors between a single standard MCMC estimator and our unbiased estimator. A standard MCMC estimator is typically slightly biased but with a smaller variance. Here, the MCMC estimator slightly overestimates the ground truth. In contrast, our unbiased estimator completely eliminates the bias but has a larger variance. For a single estimator, the standard MCMC estimator has a smaller MSE. 

Nevertheless, the benefit of no bias becomes significant in the parallel regime, as averaging over multiple processors significantly decreases the variance but keeps the bias the same. As shown in Figure \ref{fig:beta_error}, when we increase the number of processors, the relative error of our unbiased estimators eventually vanishes. In contrast, the error of the MCMC estimator will never converge to $0$ due to its systematic bias. Here the relative error from the systematic bias of MCMC is around $0.5\%$. In this example, our estimator becomes more accurate than the standard MCMC estimator when there are more than $2500$ processors. 

\begin{figure}[h] 
	\begin{subfigure}{.5\textwidth}
		\centering
		\includegraphics[width=.8\linewidth]{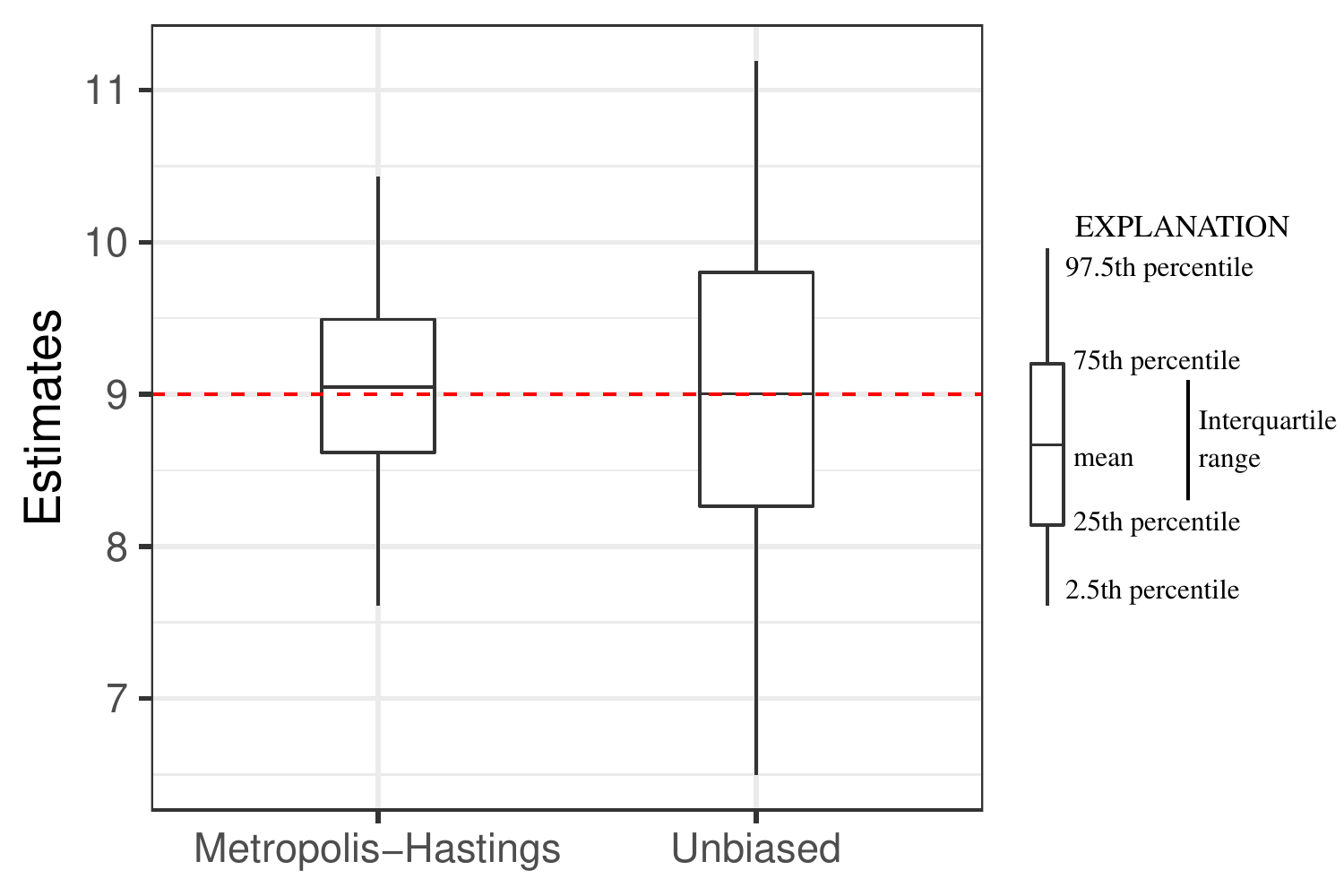}  
		\caption{Box plot of estimators for $g_K$, $K= 8$}
            \label{fig:beta_box}
	\end{subfigure}
	\hfill
	\begin{subfigure}{.5\textwidth}
		\centering
		\includegraphics[width=.8\linewidth]{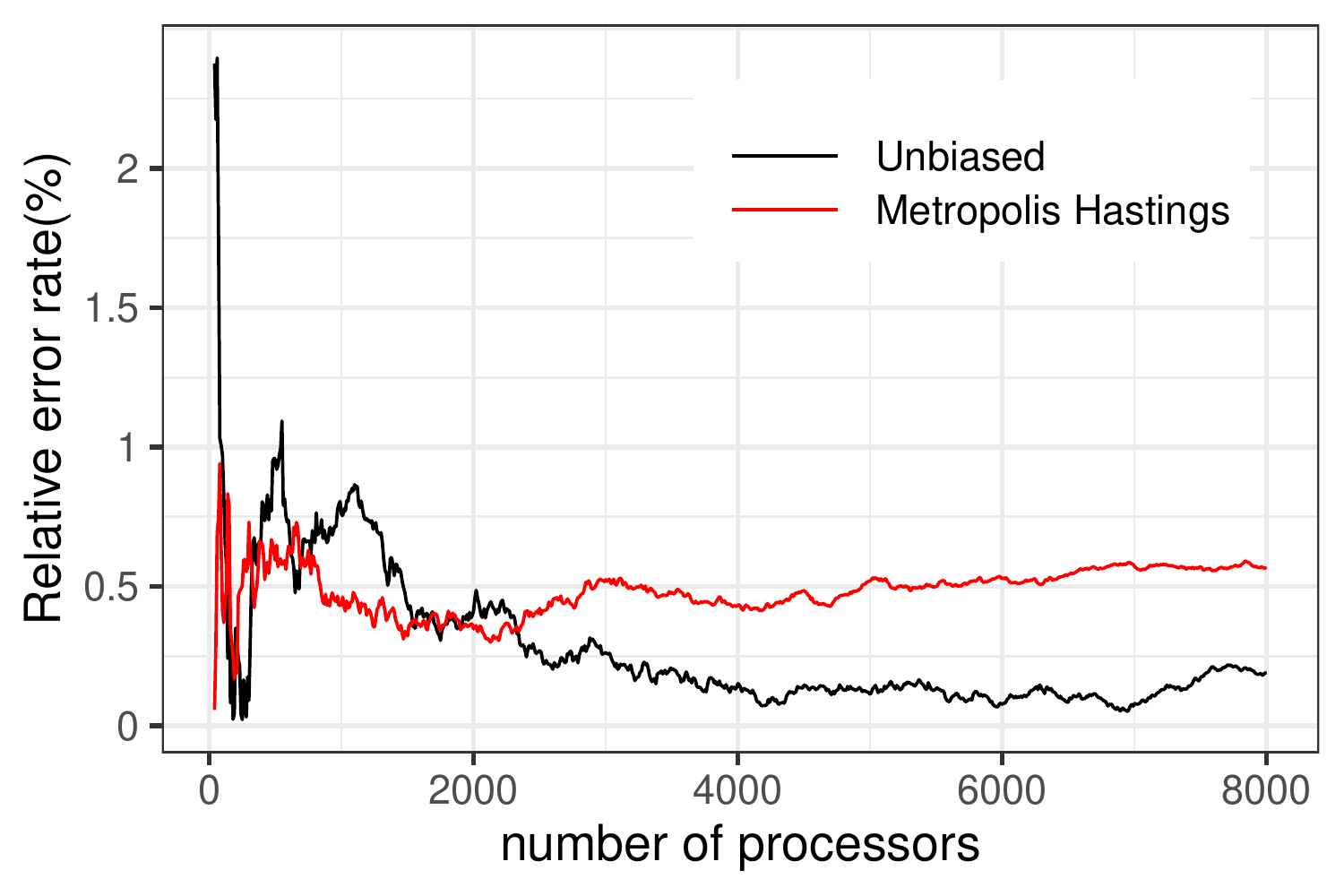}  
		\caption{Empirical relative error for $g_K$, $K = 8$}
            \label{fig:beta_error}
	\end{subfigure}
	\caption{Left: Box plot of $5\times 10^4$ estimators generated 
        by Metropolis-Hastings and Algorithm \ref{alg:MLMC}. The red dashed line represents the true value. Right: Relative error of the standard MCMC estimator (red) and unbiased estimator (black) as a function of the number of processors.}
        \label{fig:beta_compare}
\end{figure}

\subsection{Ising model}
We examine our method on the $2$-D square-lattice Ising model. Let $\Lambda$ be a set of $n\times n$ lattice sites with periodic boundary conditions. A spin configuration $\sigma\in \{-1,1\}^{n\times n}$ is an assignment of spins to all the lattice vertices. A $2$-D Ising model is a probability distribution over all the spin configurations, defined as $
p_{\theta}(\sigma) = \exp(-\theta H(\sigma))/Z(\theta)$.
Here  $H(\sigma) = -\sum_{\langle I,J\rangle}\sigma_i\sigma_j$ is the `the Hamiltonian function', where the sum is over all pairs of neighboring  sites. The normalizing constant $Z(\theta)  = \sum_\sigma \exp(-\theta H(\sigma))$ is the partition function. The parameter $\theta\geq 0$ is  interpreted as the inverse temperature in  physics. 

Now we consider the problem of estimating the ratio of normalizing constant $Z(\theta_1)/Z(\theta_2)$. The problem,  also known as estimating the free energy differences, is of great interest in computational physics and statistics  \citep{bennett1976efficient,meng1996simulating}. Since the Ising model is computationally intensive to be sampled perfectly (see \cite{propp1996exact}), unbiased estimators of $Z(\theta_1)/Z(\theta_2)$ are generally unavailable in the previous literature.  

We will use our method to construct unbiased estimators of $Z(\theta_1)/Z(\theta_2)$.
First, we notice that the ratio can be written as $
Z(\theta_1)/Z(\theta_2) = \bE_{\theta_2}[e^{\theta_2 H(\sigma)}]/\bE_{\theta_1}[e^{\theta_1 H(\sigma)}].$
For fixed $\theta_1,\theta_2$, we call the JOA estimators for unbiased estimation of $Z(\theta_1)$ and $Z(\theta_2)$ independently and feed them into Algorithm \ref{alg:MLMC} for unbiased estimators of the ratio. The JOA estimators can be obtained via coupling two Gibbs samplers using the package `unbiasedmcmc' in \cite{jacob2020unbiased}. We implement our method using  $n=12, p = 0.7, k = 4\times 10^3, m = 2k,$ $\theta_1\in \{0.02,0.03,\dots,0.18\}$ and $\theta_2 \in \{0.02,0.10\}$ on a CPU-based computer cluster. For each combination of $(\theta_1, \theta_2)$, we use our unbiased method to generate $2\times 10^4$ unbiased estimators each. We present results in Figure \ref{fig:ising_estimates}. The solid line represents our estimates for $Z(\theta_{1})/Z(0.02)$ and dash line represents our estimates for $Z(\theta_{1})/Z(0.10)$. For comparison, we also run $2\times 10^4$ independent repetitions of the standard Gibbs sampler estimators for each combination of $(\theta_1,\theta_2)$. Using the same method described in the previous example (Section \ref{subsec:product beta}), each run of the Gibbs sampler takes the same amount of time as the unbiased estimator.

To check the accuracy and compare with the standard Gibbs sampler estimator, we need to know the ground truth for every $Z(\theta_1)/Z(\theta_2)$, which is not analytically tractable. Here for each pair $(\theta_1,\theta_2)$, we run a very long Gibbs sampler for $2\times 10^5$ steps with half burn-in and run $10^4$ independent repetitions to estimate  both $\bE_{\theta_2}[e^{\theta_2 H(\sigma)}]$ and $\bE_{\theta_1}[e^{\theta_1 H(\sigma)}]$. Then we use their ratio as a proxy for our ground truth for $Z(\theta_1)/Z(\theta_2)$. Figure \ref{fig:relative_error}  compares these two methods in terms of their estimation error as a function of $\theta_1$. As shown in the plot, for every $(\theta_1\theta_2)$ pair, our unbiased estimator has a relative error very close to $0$. This suggests our estimator is highly accurate. In contrast, the Gibbs sampler has a non-negligible bias, which grows as $\theta_1$ grows. In particular, the error (which comes from bias) of the standard Gibbs sampler estimator is more than  $6\%$ when $\theta_1$ gets closer to $0.18$, while our unbiased estimator has an error much less than $1\%$.

To further examine the error of both methods as a function of the number of processors, we fix $\theta_2 = 0.1$ and choose $\theta_1 = 0.15$ and $0.18$ to plot the relative error versus the number of processors in Figure \ref{fig:normalization}. The behavior is very similar to Figure \ref{fig:beta_compare} for the Beta example. Again, as the number of processors increases, the error of the unbiased Monte Carlo estimator vanishes when the number of processors increases. In contrast, the systematic bias causes the error of the Gibbs sampler is always no less than $1.5\%$ and $6\%$ for $\theta_1 = 0.15$ and $0.18$, respectively, no matter how many processors are used. Together with the experiments in Section \ref{subsec:product beta}, it is clear that our estimator is significantly preferable to the standard Monte Carlo method when the users have many parallel processors but a limited budget per processor.

\begin{figure}[h]
	\begin{subfigure}{.5\textwidth}
		\centering
		\includegraphics[width=.8\linewidth]{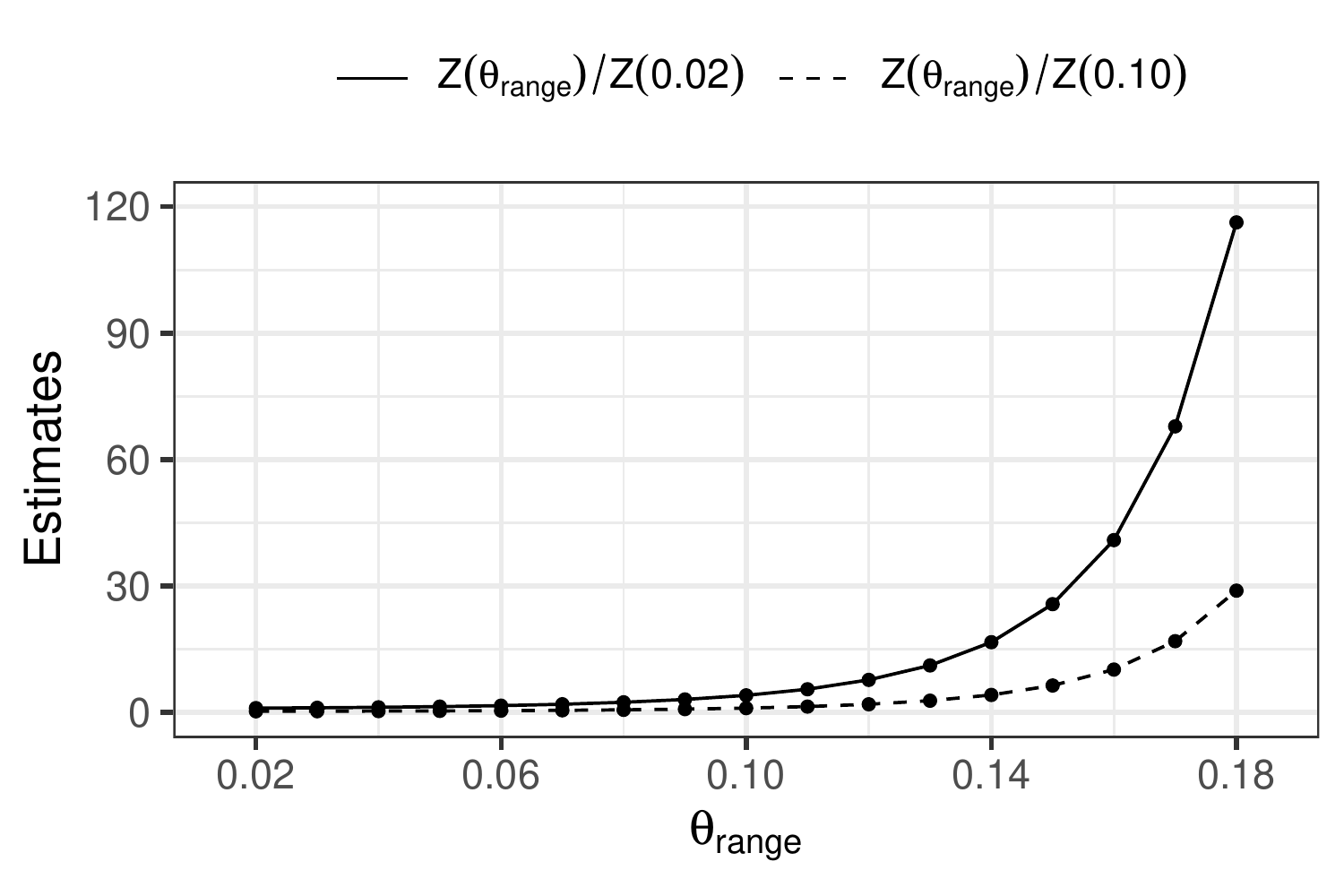}  
		\caption{Estimates for $Z_{\theta_1}/Z_{\theta_2}$ as a function of $\theta_1$}
		\label{fig:ising_estimates}
	\end{subfigure}
	\hfill
	\begin{subfigure}{.5\textwidth}
		\centering
		\includegraphics[width=.8\linewidth]{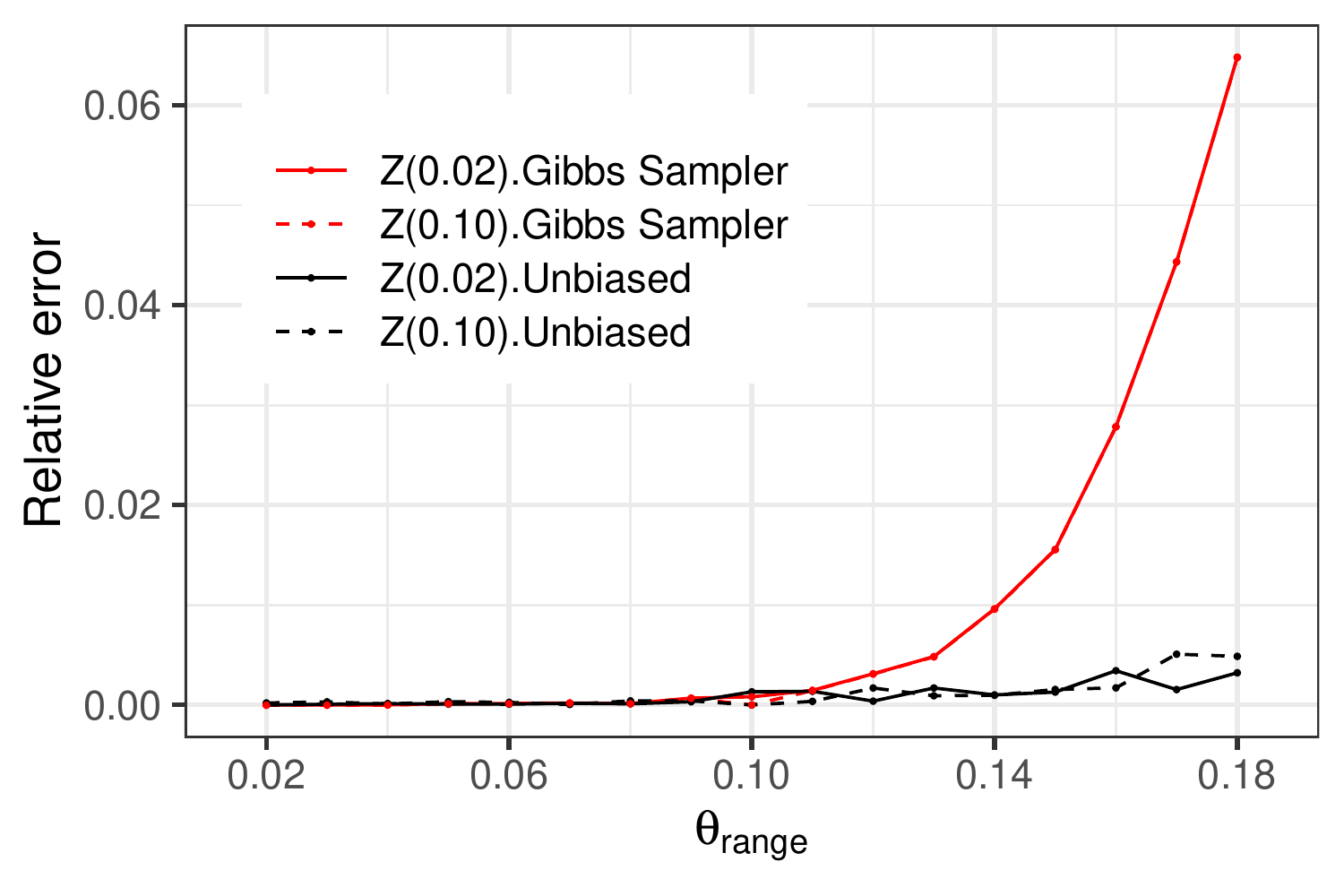}  
		\caption{Relative error of different methods as a function of $\theta_1$}
		\label{fig:relative_error}
	\end{subfigure}
	\caption{Left: The unbiased estimates of $Z(\theta_1)/Z(\theta_2)$ for $n=12$. Solid lines represent $\theta_2 = 0.02$ and dash lines represent $\theta_2 = 0.10$. Right: Relative error for different algorithms. Black lines are unbiased estimators, and red lines are standard Gibbs sampler estimators.}
	\label{fig:normalization}
\end{figure}

\begin{figure}[h]
	\begin{subfigure}{.5\textwidth}
		\centering
		\includegraphics[width=.8\linewidth]{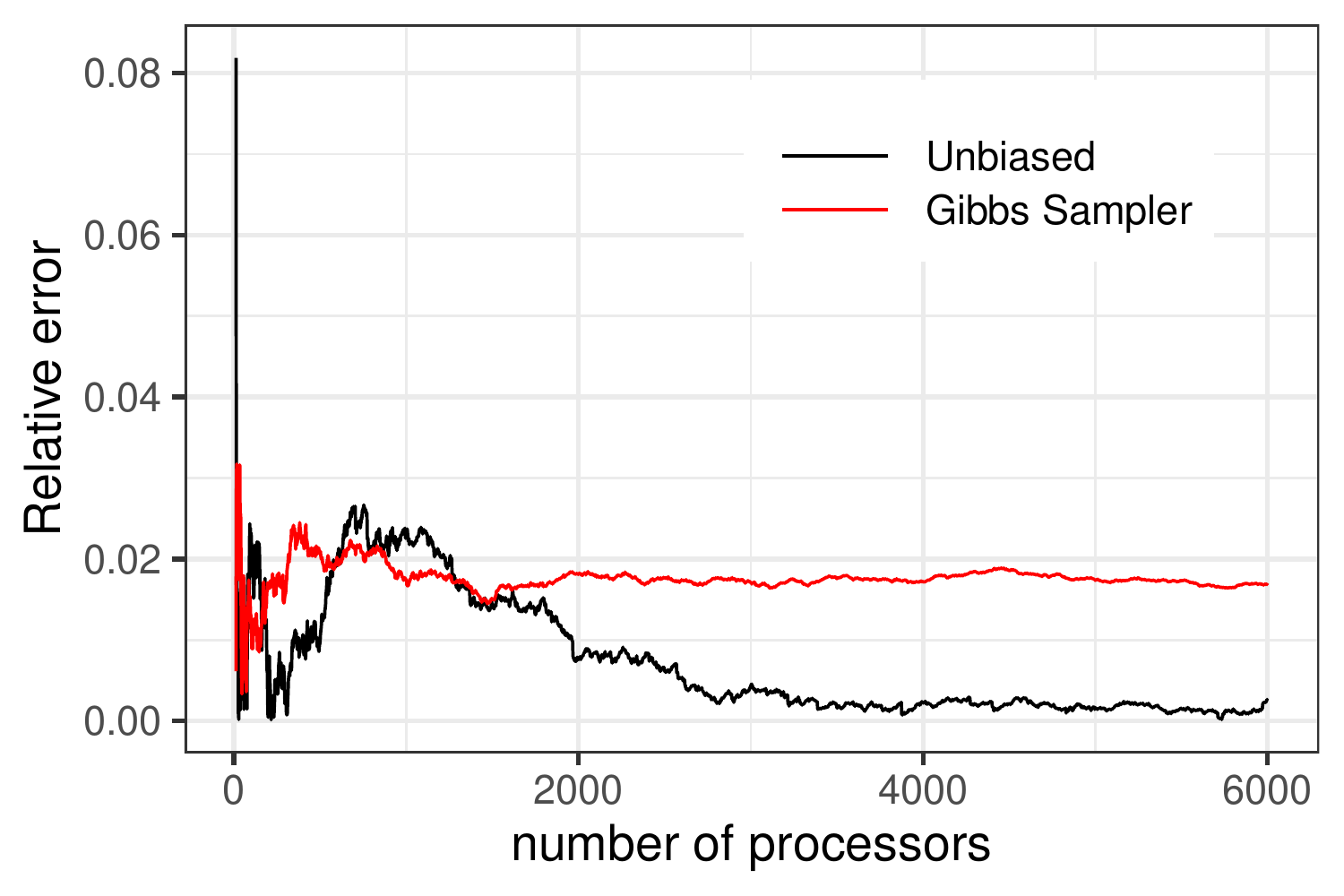}  
		\caption{Error comparison for $Z_{0.15}/Z_{0.10}$}
	\end{subfigure}
	\hfill
	\begin{subfigure}{.5\textwidth}
		\centering
		\includegraphics[width=.8\linewidth]{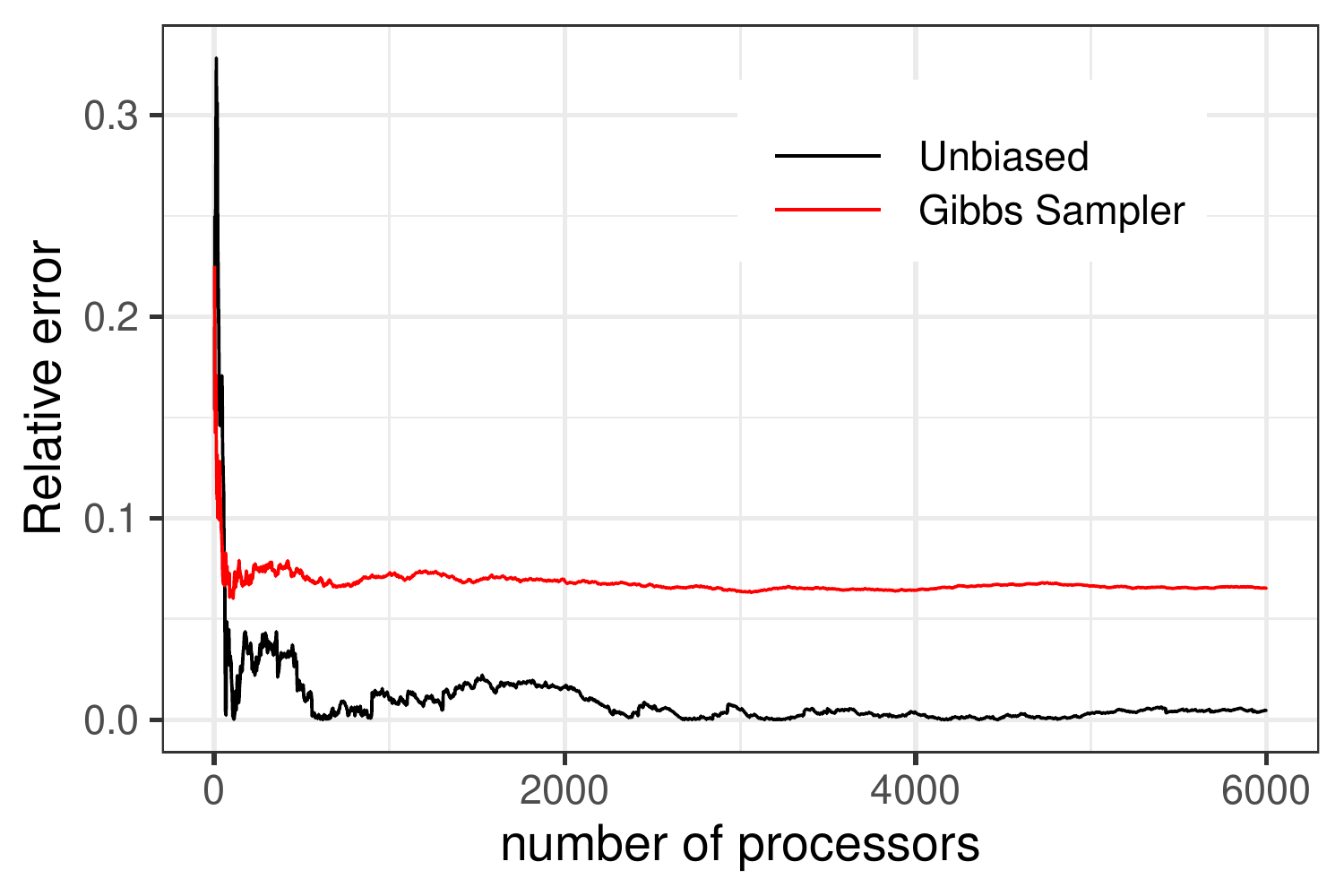}  
		\caption{Error comparison for $Z_{0.18}/Z_{0.10}$}
	\end{subfigure}
	\caption{Relative error of
the standard MCMC estimator (red) and unbiased estimator (black) as a function of the
number of processors.}
        \label{fig:ising_compare}
\end{figure}

\subsection{Nested expectation}
Finally, we estimate the following nested expectation:
$U:= \bE_{\theta_1}[\max_d \bE_{\theta_2|\theta_1}[f_d(\theta_1,\theta_2)| \theta_1]].$
The quantity $\max_d \bE_{\theta_2|\theta_1}[f_d(\theta_1,\theta_2)| \theta_1]$ is often interpreted as the utility or the optimal outcome over $D$ possible choices given the information of $\theta_1$. Since $U$ contains a nested expectation, with an out expectation over $\theta_1$ and an inner expectation over $\theta_2|\theta_1$, the vanilla Monte Carlo approach (sample $N_1$ realizations of $\theta_1$, and sample $N_2$ realizations of $\theta_2$ given each $\theta_1^{(i)}$) typically has suboptimal computational complexity $\calO(\epsilon^{-3})$ or even $\calO(\epsilon^{-4})$ for $\epsilon$ root mean square error (rMSE) under varying assumptions. Therefore, MLMC methods have been proposed when both $\theta_1$ and $\theta_2|\theta_1$ can be perfectly sampled. The case where $\theta_2|\theta_1$ can only be approximately sampled is  considered open  in \citep{giles2019decision}.

We construct unbiased estimators of $U$ using the method described in Section \ref{subsec:nest}. In this example, suppose we have two models. The first model comprises parameter $\theta_1$ with prior $\pi_1(\theta_1)$, data $Y_1$ with likelihood $p_1(y|\theta_1)$, the second model comprises  parameter $\theta_2$ with prior $\pi_2(\theta_2)$, data $Y_2$ with likelihood $p_2(y|\theta_1,\theta_2)$. The cut distribution is defined as $
\pi^\star(\theta_1,\theta_2) := \pi(\theta_1 | Y_1) \pi(\theta_2 | Y_2, \theta_1).$
This is different from the usual posterior distribution $
\pi(\theta_1,\theta_2 | Y_1,Y_2) = \pi(\theta_1 | Y_1, Y_2) \pi(\theta_2 | Y_2, \theta_1).$ In the cut model, the distribution of $\theta_1$ depends on the observations from the first model ($Y_1$) but not  the second model ($Y_2$). Since the cut model prevents the information in the second model from influencing the inference on the first, it is often used as an alternative to Bayes full posterior in the presence of model misspecification. Conducting inference on the cut model is challenging. The conditional distribution $\pi(\theta_2|Y_2,\theta_1)$ is usually  only known up a normalizing constant $Z(\theta_1)$. Standard MCMC methods on the joint space $(\theta_1,\theta_2)$ cannot be directly implemented due to the intractability of $Z(\theta_1)$, see \citep{plummer2015cuts} for detailed discussions.

In our case, we consider the real-data example used in \citep{plummer2015cuts, jacob2020unbiased} from epidemiology, which is motivated by a study of the international correlation between human papilloma virus (HPV) prevalence and cervical cancer incidence \citep{maucort2008international}. The first module consists of high-risk HPV prevalence data from $13$ countries. The data $Y_1 = \{(Z_i, N_i)\}_{i=1}^{13}$ consists of $13$ pair of integers, where $Z_i$ is the number of women infected with HPV, from country $i$ with population $N_i$. We assume a  prior $\mathsf{Beta}(1,1)$ on each component of $\theta_1$ independently, and an independent binomial likelihood $Z_i \sim \Binom(N_i,\theta_i)$ for each $i$. This yields a product beta posterior for $\theta_1$. The second module consists of the cancer data from the same $13$ countries. The data
$Y_2 = \{(X_{1,i}, X_{2,i})\}_{i=1}^{13}$ consists $13$ pair of integers, where $X_{1,i}$ is  numbers of cancer cases  arising from $X_{2,i}$ woman-years of follow-up. We assume a bivariate normal prior with mean $\mathbf 0$ and a diagonal covariance matrix with variance $10^3$ per component on the parameter $\theta_2\in \bR^2$, and a Poisson regression model $
X_{1,i} \sim \Poi(\exp(\lambda_i))$,
where  $\lambda_i = \theta_{2,1} + \theta_{1,i}\theta_{2,2} + X_{2,i}.$

Under the cut model, the first parameter $\pi(\theta_1 | Y_1)$ can be sampled from the product beta, and the second parameter can be approximately sampled from $\pi(\theta_2|Y_2,\theta_1)$ using MCMC. Suppose we are interested in $
U:= \bE_{\theta_1}[\max_{d\in\{1,2,\ldots, 13\}} \bE_{\theta_2|\theta_1}[\lambda_d]],$
which corresponds to the expectation of the largest parameter in the Poisson regression after observing $\theta_1$. We implement Algorithm \ref{alg:MLMCnest} with parameter $p = 0.7$ to get unbiased estimators of $U$. In each run, we first sample one $\theta_1$ from the product beta posterior,  then use the JOA estimator with $k = 2\times 10^3, m = 3\times 10^3$ by the R package `unbiasedMCMC'   to generate unbiased estimators of   $\bE_{\theta_2|\theta_1}[\lambda_d]$. Finally, we use the unbiased MLMC method to eliminate the bias. Our estimates are presented in Figure \ref{fig:Cut} below. Figure \ref{fig:Cut}(a) gives the estimates and their CIs of $\lambda_d$ for each $d$. Figure \ref{fig:Cut}(b) gives the histogram and the fitted curve from $10^5$ unbiased estimators of $U$. Figure \ref{fig:Cut}(a) suggests the $12$-th country has the largest  $\lambda_d$, which is around $21$, which is consistent with the result from our unbiased estimator on Figure \ref{fig:Cut}(b).
\begin{figure}
	\begin{subfigure}{.5\textwidth}
		\centering
		\includegraphics[width=.8\linewidth]{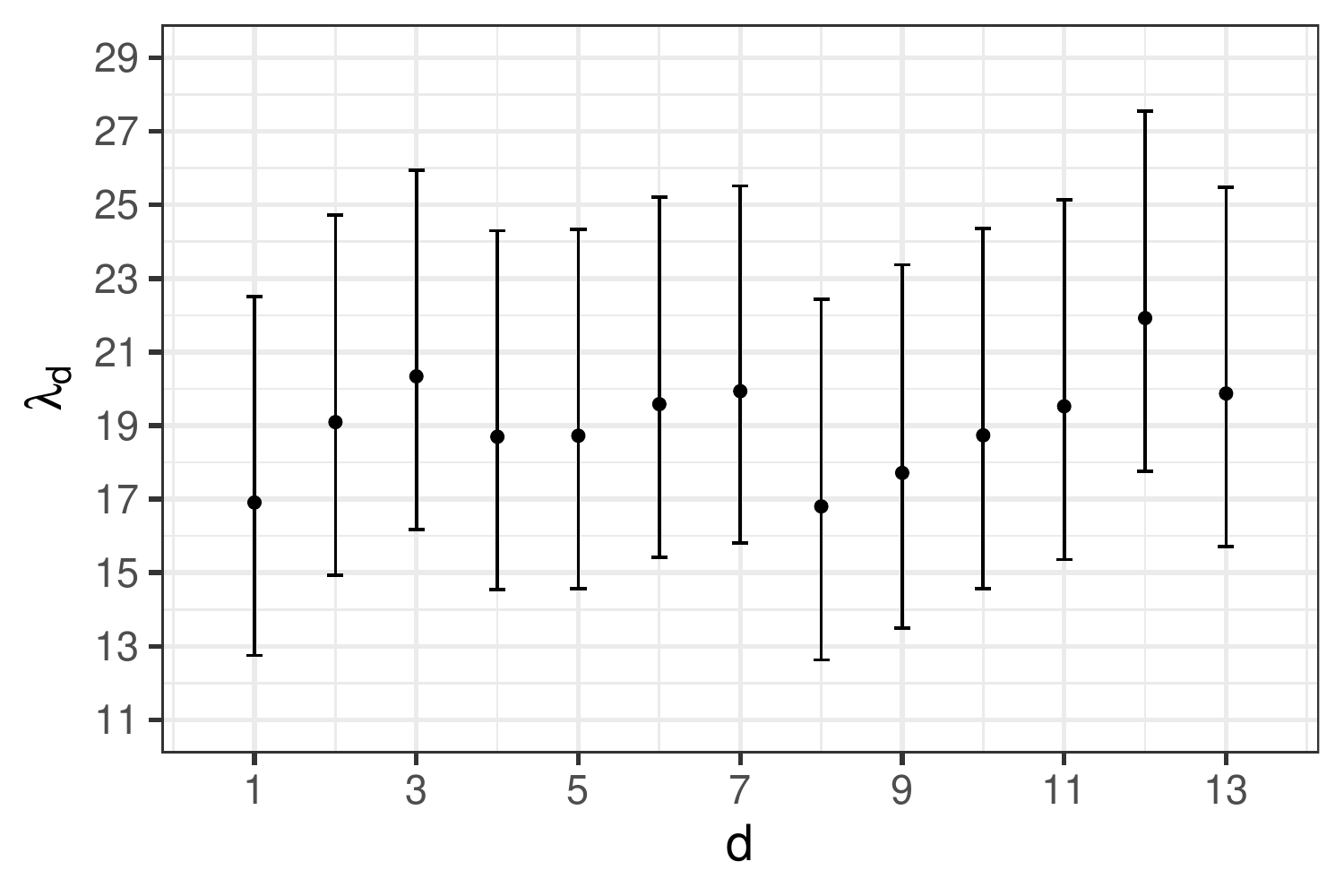}  
		\caption{Estimates and 95\% confidence intervals for $\lambda_d$, computed from $10^5$ JOA estimators.}
		\label{fig:cut_1}
	\end{subfigure}
	\begin{subfigure}{.5\textwidth}
		\centering
		\includegraphics[width=.8\linewidth]{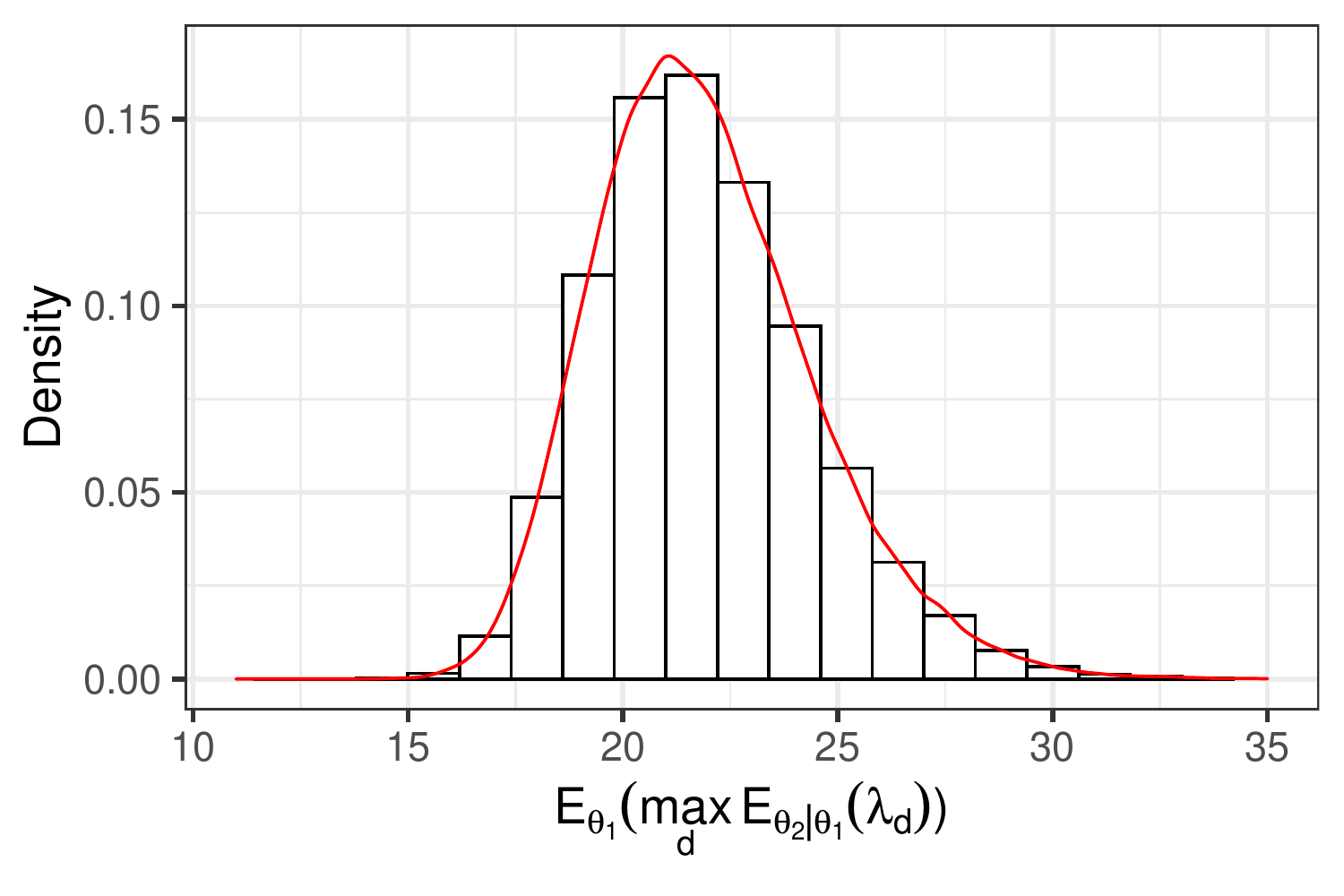}  
		\caption{Histogram of $\bE_{\theta_1}[\max_{d\in\{1,2,\ldots, 13\}} \bE_{\theta_2|\theta_1}[\lambda_d]]$ computed from $10^5$ calls of Algorithm \ref{alg:MLMCnest}. }
		\label{fig:cut_2}
	\end{subfigure}
	\caption{}
	\label{fig:Cut}
\end{figure}

\section{Future works}\label{sec:conclusion}
Based on the combination and generalization of the unbiased MCMC and MLMC method, we propose general unbiased estimators of $g(\bE_\pi[f])$ when $\pi$ can only be approximately sampled.  We further extend this framework to estimate nested expectations under intractable distributions. Although promising,  the existing framework (Algorithm \ref{alg:MLMC} and its variants) still has the potential to be generalized.  We highlight the potential paths forward.

First,   $\calT$ is assumed to be a function of the expectation. This assumption excludes many important applications, including the quantile and maximum a posteriori (MAP) estimations,  where $\calT$ depends directly on the probability measure instead of the expectation of some probability measure. We plan to develop a general method to include some/all of the applications above. Taking a step back, many computational challenges remain even assuming $\calT(\pi) := g(\bE_\pi[f])$. Algorithm \ref{alg:MLMC} implicitly requires the range of $S_H(m)/m$ is a subset of the domain of $g$. For example, our algorithm fails when $g(x) = \sqrt{x}$ since the JOA estimator may not always be non-negative. As remarked by several authors \citep{lyne2015russian}, the domain problem is deeply connected with the sign problem in computational physics, which is NP-hard in its general form. Progress on the domain problem should not only let us improve our existing framework but also benefit both the statistics and physics communities. Lastly,  the efficiency of the existing estimator (Algorithm \ref{alg:MLMC}) is still pretty much unexplored. In practice, we find the implementation time can be slow when the dimension is high, or the Markov chain mixes slowly. In particular, empirical results suggest that the parameter $p$ in Algorithm \ref{alg:MLMC} significantly influences both the variance and the computation cost. Therefore, finding the optimal parameter and the tradeoff between computational and statistical efficiency is an interesting problem.

\bibliographystyle{chicago}
\setstretch{1.24}
\bibliography{bib}
\newpage
\appendix
\section{Proofs}\label{sec:appendix}
\subsection{Auxiliary  Lemmas}\label{subsec: auxiliary lemma}
In this section we prove some auxiliary results that will be used throughout the technical proofs. We start (without proof) the well-known Marcinkiewicz-Zygmund inequality, and then prove two useful corollaries based on this inequality.
\begin{lemma}[Marcinkiewicz-Zygmund inequality \citep{marcinkiewicz1937quelques}]\label{lemma.M-Z_inequality}
	If $X_1,\cdots, X_n$ are independent random variables with $\bE[X_i] = 0$ and $\bE\left[|X_i|^p\right]<\infty$ for some $p>2$. Then,
	\begin{equation*}
		\bE \left[\left|\sum_{i=1}^nX_i\right|^p\right]\leq C_p\bE \left[\left(\sum_{i=1}^n|X_i|^2\right)^{p/2}\right],
	\end{equation*}
	where $C_p$ is a constant that only depends on $p$. 
\end{lemma}

One  corollary of the Marcinkiewicz-Zygmund inequality is:
\begin{corollary}\label{cor:M-Z iid}
	
	With all the assumptions as above, if we further assume that $X_1, \cdots, X_n$ are i.i.d. . Then,
	\begin{equation*}
		\bE \left[\left|\frac{1}{n}\sum_{i=1}^nX_i\right|^p\right]\leq  C_p\frac{\bE|X_1|^p}{n^{p/2}}
	\end{equation*}
	for every $p\geq 2$.
\end{corollary}
\begin{proof}[Proof of Corollary \ref{cor:M-Z iid}]
	Applying the Marcinkiewicz-Zygmund inequality on \\$(X_1/n, X_2/n,\ldots, X_n/n)$, we have:
	\begin{equation*}
		\bE \left[\left|\frac{1}{n}\sum_{i=1}^nX_i\right|^p\right]\leq C_p \bE\left[\left(\sum_{i=1}^n \big\lvert\frac{ X_i}{n}\big\rvert^2\right)^{p/2}\right].
	\end{equation*}
	Since $x^{p/2}$ is convex, we have 
	\begin{equation*}
		\left(\sum_{i=1}^n \big\lvert\frac{ X_i}{n}\big\rvert^2\right)^{p/2} = \left(\frac 1n \sum_{i=1}^n\frac{ \lvert X_i \rvert^2}{n}\right)^{p/2} \leq \frac 1n \sum_{i=1}^n \frac{\lvert X_i\rvert^p}{n^{p/2}}.
	\end{equation*}
	Taking expectation on both sides of the above inequality yields
	\begin{equation*}
		\bE\left[\left(\sum_{i=1}^n \big\lvert\frac{ X_i}{n}\big\rvert^2\right)^{p/2} \right] \leq \frac{\bE|X_1|^p}{n^{p/2}},
	\end{equation*}
	and our desired inequality follows.
\end{proof}

The Marcinkiewicz-Zygmund inequality naturally generalizes to random vectors.
\begin{corollary}[Multivariate Marcinkiewicz-Zygmund inequality]\label{cor: multivariate M-Z}
	Let $X_1, \cdots, X_n$ be i.i.d. random vectors in $\bR^m$, with $\bE[X_i] = \mathbf{0}$ and $\bE[\lVert X_i \rVert_p^p] = \bE[\sum_{j=1}^m \lvert X_{i,j}\rvert^p] < \infty$. Then
	\begin{equation*}
		\bE \left[\big\lVert\frac{1}{n}\sum_{i=1}^nX_i\big\rVert_p^p\right]\leq  C_p\frac{\bE\left[\lVert X_1\rVert_p^p\right]}{n^{p/2}}
	\end{equation*}
	for every $p\geq 2$.
\end{corollary}
\begin{proof}[Proof of Corollary \ref{cor: multivariate M-Z}]
	We know $$\bE \left[\big\lVert\frac{1}{n}\sum_{i=1}^nX_i\big\rVert_p^p\right] = \sum_{j=1}^m\bE\left[ \big\lvert \frac 1n\sum_{i=1}^n X_{i,j}\big\rvert^p\right].$$
	Applying Corollary \ref{cor:M-Z iid} on each component of each $X_i$ yields
	$$
	\sum_{j=1}^m\bE\left[ \big\lvert \frac 1n\sum_{i=1}^n X_{i,j}\big\rvert^p\right]\leq C_p \sum_{j=1}^m \frac{\bE\lvert X_{1,j}\rvert^p}{n^{p/2}} = C_p\frac{\bE\left[\lVert X_1\rVert_p^p\right]}{n^{p/2}},
	$$
	as desired.
\end{proof}
We also need the following inequality to compare $\lVert x\rVert_p$ and $\lVert x \rVert_q$ for $p\neq q$ and $x\in \bR^m$. The proof follows directly from the H\"{o}lder's inequality. 
\begin{lemma}\label{lem:norm-comparison}
	For any $x\in \bR^m$ and $p < q$, we have:
	\[
	\lVert x \rVert_p \leq m^{1/p - 1/q} \lVert x \rVert_q.
	\]
\end{lemma}
\begin{proof}
	\begin{align*}
		\lVert x \rVert_p^p = \sum_{i=1}^m \lvert x_i\rvert ^p\cdot 1\leq  \left(\sum_{i=1}^m \lvert x_i\rvert ^q\right)^{p/q} m^{1-p/q}
	\end{align*}
	where the last inequality follows from the H\"{o}lder's inequality.  Our result follows by taking the $(1/p)$-th power on both sides. 
\end{proof}

\subsection{Bounding $\bE[\lvert \Delta_n\rvert^2]$}\label{subsec: bounding delta}
Recall that  $\Delta_n = g\left(S_H(2^{n})/2^n\right) -\frac{1}{2}\left(g\left(S_H^O(2^{n-1})/2^{n-1}\right) + g\left(S_H^E(2^{n-1})/2^{n-1}\right)\right)$, and the final estimator takes the form $\Delta_N/p_N + g(H_1)$. Therefore, understanding the theoretical properties of $\Delta_n$ is crucial for studying our estimator.  

\begin{proof}[Proof of Lemma \ref{lem: delta-square}]
	For simplicity, we denote $m(\pi)$ by $\mu$. By Assumption \ref{ass: differentiable},  there exists $\epsilon>0$
	such that $g$ is $\alpha$-H\"{o}lder continuous on $(\mu - \epsilon, \mu + \epsilon)$, we can then write $\Delta_n$ as:
	\begin{align}\label{eqn:Delta, split}
		\lvert \Delta_n\rvert = \lvert \Delta_n\rvert \Indc(A_1) +  \lvert \Delta_n\rvert \Indc(A_2),
	\end{align}
	where $A_1$ is the event  $$
	\left\{\left\lVert\frac{S_H^\sfO(2^{n-1})}{2^{n-1}} - \mu\right\rVert <\epsilon\right\} \cap \left\{\left\lVert\frac{S_H^\sfE(2^{n-1})}{2^{n-1}} - \mu\right\rVert < \epsilon \right\},
	$$
	and $A_2$ is the event
	$$\left\{\max\left(\left\lVert\frac{S_H^\sfO(2^{n-1})}{2^{n-1}} - \mu\right\rVert,\left\lVert\frac{S_H^\sfE(2^{n-1})}{2^{n-1}} - \mu\right\rVert \right)\geq\epsilon\right\}.$$ 
	
	Under the event $A_1$, we have  $\left\lVert\frac{S_H^\sfO(2^{n-1})}{2^{n-1}} - \mu\right\rVert <\epsilon$ and $\left\lVert\frac{S_H^\sfE(2^{n-1})}{2^{n-1}} - \mu\right\rVert <\epsilon$. This further implies  $$\left\lVert\frac{S_H(2^{n})}{2^{n}} - \mu\right\rVert < \epsilon$$
	
	by the triangle inequality and the fact  $\frac{S_H(2^{n})}{2^{n}} = \frac12  \left(\frac{S_H^\sfO(2^{n-1})}{2^{n-1}} +  \frac{S_H^\sfE(2^{n-1})}{2^{n-1}}\right).$
	
	Then we can write $\Delta_n$ as:
	\begin{align*}
		\Delta_n & = g\left(\frac{S_H(2^{n})}{2^{n}}\right) -  \frac12  \left(\frac{S_H^\sfO(2^{n-1})}{2^{n-1}} +  \frac{S_H^\sfE(2^{n-1})}{2^{n-1}}\right)\\
		& = \frac{1}{2}\left(g\left(\frac{S_H(2^{n})}{2^{n}}\right) - g\left(\frac{S_H^\sfO(2^{n-1})}{2^{n-1}}\right) \right) + \frac{1}{2}\left(g\left(\frac{S_H(2^{n})}{2^{n}}\right) - g\left(\frac{S_H^\sfE(2^{n-1})}{2^{n-1}}\right) \right) \\
		& = \frac{1}{2}Dg(\xi^\sfO_n) \left(\frac{S_H(2^{n})}{2^{n}} -\frac{S_H^\sfO(2^{n-1})}{2^{n-1}}\right) + \frac{1}{2}Dg(\xi^\sfE_n) \left(\frac{S_H(2^{n})}{2^{n}} -\frac{S_H^\sfE(2^{n-1})}{2^{n-1}}\right) \\
		& = \frac{1}{4}\left(Dg(\xi^\sfO_n)-Dg(\xi^\sfE_n)\right)\frac{S_H^\sfE(2^{n-1}) - S_H^\sfO(2^{n-1}) }{2^{n-1}},
	\end{align*}
	
	where $\xi^\sfO_n$ is a convex combination of $\frac{S_H(2^{n})}{2^{n}}$ and $\frac{S_H^\sfO(2^{n-1})}{2^{n-1}}$, $\xi^\sfE_n$ is a convex combination of $\frac{S_H(2^{n})}{2^{n}}$ and $\frac{S_H^\sfE(2^{n-1})}{2^{n-1}}$ by the Multivariate Mean value Theorem. Under $A_1$,  both $\xi^\sfO_n$  and $\xi^\sfE_n$ are within the $\epsilon$-neighbor of $\mu$, applying the $\alpha$-H\"{o}lder continuous assumption yields
	
	\begin{align*}
		\lvert\Delta_n\rvert \leq c_1(\epsilon)\left\lVert\xi_n^\sfO-\xi_n^\sfE\right\rVert^{\alpha}\cdot\left\lVert \frac{S_H^\sfO(2^{n-1}) - S_H^\sfE(2^{n-1}) }{2^{n-1}}\right\rVert \leq c_2(\epsilon)\left\lVert \frac{S_H^\sfO(2^{n-1}) - S_H^\sfE(2^{n-1}) }{2^{n-1}} \right\rVert^{1+\alpha}.
	\end{align*}
	
	Then, 
	\begin{align}\label{eqn: Delta_n on A1}
		&\bE\left[\lvert\Delta_n\rvert^2 \Indc(A_1) \right]
		\leq c_2(\epsilon)\mathbb{E}\left[\left\lVert \frac{S_H^\sfO(2^{n-1}) - S_H^\sfE(2^{n-1}) }{2^{n-1}} \right\rVert^{2(1+\alpha)}\right].
	\end{align}
	Since $S_H^\sfO(2^{n-1})$ and $S_H^\sfE(2^{n-1})$ are vectors in $\bR^m$, applying 
	Lemma  \ref{lem:norm-comparison} on $p = 2, q = 2(1+\alpha)$ gives:
	\begin{align}\label{eqn: Delta_n Holder}
		\left\lVert \frac{S_H^\sfO(2^{n-1}) - S_H^\sfE(2^{n-1}) }{2^{n-1}} \right\rVert^{2(1+\alpha)} \leq m^\alpha \left\lVert \frac{S_H^\sfO(2^{n-1}) - S_H^\sfE(2^{n-1}) }{2^{n-1}} \right\rVert^{2(1+\alpha)}_{2(1+\alpha)}
	\end{align}

	Since $S_H^\sfO(2^{n-1}) - S_H^\sfE(2^{n-1})$ is the sum of $2^{n-1}$ i.i.d. random variables, each with the same distribution as $H_2 - H_1$, applying the Multivariate Marcinkiewicz-Zygmund  inequality (Corollary \ref{cor: multivariate M-Z}) gives us:
	\begin{align}
		\mathbb{E}\left[\left\lVert \frac{S_H^\sfO(2^{n-1}) - S_H^\sfE(2^{n-1}) }{2^{n-1}} \right\rVert^{2(1+\alpha)}_{2(1+\alpha)}\right] \leq C_{2(1+\alpha)}\cdot \frac{\bE\left[\lVert H_2 - H_1\rVert^{2(1+\alpha)}_{2(1+\alpha)}\right]}{2^{(1+\alpha)(n-1)}}\\
		\label{eqn: upper bound on A1}
		\leq C_{2(1+\alpha)}\cdot2^{3(1+\alpha)}\cdot \frac{\bE\left[\lVert  H_1\rVert^{2(1+\alpha)}_{2(1+\alpha)}\right]}{2^{(1+\alpha)n}}.
	\end{align}
	where the last step uses  the  inequality $(a+b)^p \leq 2^{p-1} (\lvert a \rvert^p + \lvert b \rvert^p)$ for $p\geq 2$. It is worth mentioning that the right hand side of \eqref{eqn: upper bound on A1} is finite as Assumption \ref{ass: moment} guarantees $H_1$ has finite $l$-th moment with $l > 2+\alpha$.  Combining \eqref{eqn: Delta_n on A1}, \eqref{eqn: Delta_n Holder}, and \eqref{eqn: upper bound on A1}, we have:
	\begin{align}\label{eqn:A1, final}
		\bE\left[\lVert\Delta_n\rVert^2 \Indc(A_1)\right]\leq C_1(m, \alpha,\epsilon) 2^{-n(1+\alpha)},
	\end{align}
	where $C_1(m,\alpha,\epsilon) = c_2(\epsilon)\cdot C_{2(1+\alpha)} \cdot 2^{3(1+\alpha)}\cdot {\bE\left[\lVert  H_1\rVert^{2(1+\alpha)}_{2(1+\alpha)}\right]}$ is a constant when Assumption \ref{ass: domain} -- \ref{ass: moment} are satisfied.
	
	Under $A_2$, we have:
	\begin{align}\label{A2,split}
		\lvert \Delta_n\rvert^2 \Indc(A_2) \leq  \lvert \Delta_n\rvert^2 \Indc \left(\left\lVert\frac{S_H^\sfO(2^{n-1})}{2^{n-1}} - \mu\right\rVert >\epsilon \right) + \lvert \Delta_n\rvert^2 \Indc\left(\left\lVert\frac{S_H^\sfE(2^{n-1})}{2^{n-1}} - \mu\right\rVert > \epsilon\right)
	\end{align}

	Now we upper bound the first term's expectation, 
	\begin{align}
		\mathbb{E}\left[ \lvert\Delta_n\rvert^2\Indc\big(\left\lVert \frac{S_H^\sfO(2^{n-1})}{2^{n-1}} -\mu\right\rVert>\epsilon\big)\right]
		\label{eqn:A2 Holder}&\leq \bE[\lvert \Delta_n\rvert^{2s}]^{1/s} \mathbb P\left(\left\lVert\frac{S_H^\sfO(2^{n-1})}{2^{n-1}} -\mu\right\rVert > \epsilon \right)^{(s-1)/s}\\
		\label{eqn:A2 tradeoff}&\leq \calC_s^{1/s} 2^{-\alpha_sn/s}  \bP\left(\left\lVert\frac{S_H^\sfO(2^{n-1})}{2^{n-1}} -\mu\right\rVert > \epsilon \right)^{(s-1)/s}\\
		\label{eqn:A2  Markov}&\leq  \calC_s^{1/s}  \cdot (\epsilon^{-l(s-1)/s}) \cdot 2^{-\alpha_sn/s} \cdot \bE\left[\big\lVert\frac{S_H^\sfO(2^{n-1})}{2^{n-1}} - \mu\big\rVert^l\right]^{(s-1)/s}.
	\end{align}
	Here \eqref{eqn:A2 Holder} follows from the H\"{o}lder's inequality, \eqref{eqn:A2 tradeoff} uses Assumption \ref{ass:tradeoff}, and \eqref{eqn:A2  Markov} follows from the Markov's inequality. 
	Again, using Lemma \ref{lem:norm-comparison} and Corollary \ref{cor: multivariate M-Z}, the term $\bE\left[\big\lVert\frac{S_H^\sfO(2^{n-1})}{2^{n-1}} - \mu\big\rVert^l\right]$ can be upper bounded by:
	\begin{align}\label{eqn:A2 CZ}
		\bE\left[\big\lVert\frac{S_H^\sfO(2^{n-1})}{2^{n-1}} - \mu\big\rVert^l\right]\leq m^{l/2 - 1}\bE\left[\big\lVert\frac{S_H^\sfO(2^{n-1})}{2^{n-1}} - \mu\big\rVert^l_l\right]\leq 2^{l/2}\cdot m^{l/2 - 1}\cdot C_l \cdot \frac{\bE\left[\lVert H_1 \rVert^l_l\right]}{2^{nl/2}}.
	\end{align}
	Combining \eqref{eqn:A2  Markov} and \eqref{eqn:A2 CZ}, we have
	\begin{align*}
		\mathbb{E}\left[ \lvert\Delta_n\rvert^2\Indc\big(\left\lVert \frac{S_H^\sfO(2^{n-1})}{2^{n-1}} -\mu\right\rVert>\epsilon\big)\right] &\leq C_2(m, l, \epsilon,s) 2^{-\alpha_sn/s}  2^{-nl(s-1)/(2s)} \\
		&= C_2(m, l, \epsilon,s) 2^{-n\big(\frac{\alpha_s}{s} + \frac{(s-1)l}{2s}	\big)},
	\end{align*}
	where 	 $C_2(m, l, \epsilon,s) = \calC_s^{1/s} \cdot \left(\epsilon^{-l}2^{l/2}\cdot m^{l/2 - 1}\cdot C_l\cdot \bE\left[\lVert H_1 \rVert^l_l\right]\right)^{(s-1)/s} $ is a constant when Assumption \ref{ass: domain} -- \ref{ass:tradeoff} are satisfied. Furthermore, by  Assumption \ref{ass:tradeoff}, $2\alpha_s + (s-1)l > 2s$. It is clear that $\frac{\alpha_s} {s}+ \frac{(s-1)l}{2s}  > 1$, and therefore 
	\begin{align}\label{eqn:A2, odd}
		\mathbb{E}\left[ \lvert\Delta_n\rvert^2\Indc\big(\left\lVert \frac{S_H^\sfO(2^{n-1})}{2^{n-1}} -\mu\right\rVert>\epsilon\big)\right]\leq C_2(m,l,\epsilon,s) 2^{-(1+\tilde\alpha) n},
	\end{align}
	where $\tilde\alpha  = \frac{\alpha_s}{s} + \frac{(s-1)l}{2s} - 1 >0$. The same argument also shows
	\begin{align}\label{eqn: A2, even}
		\mathbb{E}\left[ \lvert\Delta_n\rvert^2\Indc\big(\left\lVert \frac{S_H^\sfO(2^{n-1})}{2^{n-1}} -\mu\right\rVert>\epsilon\big)\right]\leq C_2(m,l,\epsilon,s) 2^{-(1+\tilde\alpha) n}.
	\end{align}
	Combining \eqref{eqn:A2, odd}, \eqref{eqn: A2, even}, and \eqref{A2,split}, we have
	\begin{align}\label{eqn: A2, bound}
		\mathbb{E}\left[ \lvert\Delta_n\rvert^2\Indc(A_2)\right]\leq 2C_2(m,l,\epsilon,s) 2^{-(1+\tilde\alpha) n}.
	\end{align}
	Finally, taking $\gamma = \min\{\alpha,\tilde\alpha\}$, $C = C_1 + 2C_2$, and using  \eqref{eqn:Delta, split}, \eqref{eqn:A1, final}, and \eqref{eqn: A2, bound}, we conclude:
	\begin{align}\label{eqn:Delta final}
		\mathbb E[\lvert \Delta_n\rvert ^2] \leq C 2^{-n(1+\gamma)}.
	\end{align}
\end{proof}

\subsection{The Moment Assumption \ref{ass: moment} and Markov chain mixing}\label{subsec: moment assumption and mixing}
In this subsection we discuss the relation between  the  Moment Assumption \ref{ass: moment}  and the mixing time of the underlying Markov chain. Throughout this subsection, the unbiased estimator $H$ of $m(\pi)$ is assumed to be the JOA estimator $H_k(Y,Z)$ defined in Section \ref{subsubsec:unbiased MCMC}, which also extends to $H_{k:m}(Y,Z) = (m-k+1)^{-1}\sum_{l=k}^m H_l(Y,Z)$ naturally.

Before giving a formal statement of Proposition \ref{prop: moment, informal}, we first recall some definitions in Markov chain theory. We say a $\pi$-invariant, $\phi$-irreducible and aperiodic Markov transition kernel $P$ satisfies a geometric drift condition if there exists a measurable function $V: \Omega \rightarrow [1,\infty)$, $\lambda\in (0,1)$, and a measurable set $\calS$ such that for all $x\in\Omega$:
\begin{align}\label{eqn:drift}
	\int P(x,\diff y) V(y) \leq \lambda V(x) + b\Indc(x\in \calS).
\end{align}
Moreover, the set $\calS$ is called a small set if there exists a positive integer $m$, $\epsilon > 0$, and a probability measure $\nu$ on such that for every $x\in \calS$:
\begin{align}\label{eqn: small}
	P^m(x,\cdot) \geq  \epsilon \mu(\cdot).
\end{align}
The technical definitions for irreducibility, aperiodicity and small sets can be found in Chapter 5 of \cite{meyn2012markov}. The geometric drift condition is  a key tool guaranteeing the geometric ergodicity of a Markov chain, meaning the Markov chain $P$ converges to its stationary distribution $\pi$ at a geometric rate. It is known that the geometric drift condition is satisfied for a wide family of Metropolis-Hastings algorithms. We refer the readers to \cite{mengersen1996rates, roberts1996geometric} for existing results. 

Now we give a formal statement of Proposition \ref{prop: moment, informal}.
\begin{prop}[Verifying Assumption \ref{ass: moment}, formal version of Proposition \ref{prop: moment, informal}]\label{prop: moment, formal}
	Suppose the Markov transition kernel described in Section \ref{subsubsec:unbiased MCMC} satisfies a geometric drift condition with a small set $\calS$ of the form $\calS =\{x: V(x)\leq L\}$ for $\lambda + b/(1+L) < 1$. Suppose there exists $\tilde \epsilon \in (0,1)$ such that
	\[
	\inf_{(x,y)\in \calS \times \calS} \barP((x,y), \calD) \geq \tilde\epsilon,
	\]
	where $\calD:= \{(x,x): x\in \Omega\}$ is the diagonal of $\Omega\times \Omega$. Suppose  also there exists $p > l$ and $D_p > 0$ such that $\bE[\Vert f(Y_t)\rVert_p^p] < D_p$ for every $t$. 	Then $\mathbb{E}[\lVert H_k(Y,Z)\rVert_l^l] <\infty$ for every $k$.  
\end{prop}

The main ingredient in the proof of Proposition \ref{prop: moment, formal} is to control the tail probability of the meeting time $\tau$. We say $\tau$ has a $\beta$-polynomial tail if there exists a constant $K_\beta > 0$ such that
\begin{align}\label{eqn: poly tail}
	\bP(\tau > n)\leq K_\beta n^{-\beta}.
\end{align}
We say $\tau$ has an exponential tail if there exists a constant $K > 0$ and $\gamma\in(0,1)$ such that
\begin{align}\label{eqn: expo tail}
	\bP(\tau > n)\leq K\gamma^n.
\end{align}

Our next result gives sufficient conditions to ensure Assumption \ref{ass: moment}.

\begin{lemma}\label{lem: tail-moment}
	Suppose one of the following holds:
	\begin{itemize}
		\item There exists $p> l$, $\beta> 0$, and $D_p> 0$ such that $\frac 1p + \beta  > \frac 1l$; $\bE[\lVert f(Y_t)\rVert_p^p] < D_p$ for every $t$, and $\tau$ has a $\beta$-polynomial tail; 
		\item There exists $p > l$ and $D_p > 0$ such that $\bE[\Vert f(Y_t)\rVert_p^p] < D_p$ for every $t$,  and $\tau$ has an exponential tail. 
	\end{itemize}
	Then $\mathbb{E}[\lVert H_k(Y,Z)\rVert_l^l] <\infty$ for every $k$. 
\end{lemma}
\begin{proof}[Proof of Lemma \ref{lem: tail-moment}]
	We start with the first case. Without loss of generality, we assume $k= 0$ and  the estimator $H_0(Y,Z):= f(Y_0) + \sum_{i=1}^{\tau-1}(f(Y_i) - f(Z_{i-1}))$ takes scalar value. Let $D_k := f(Y_k) - f(Z_{k-1})$ for $k \geq 1$, and $D_0 = f(Y_0)$, the estimator can be written as:
	\begin{align*}
		H_0(Y,Z) = \sum_{k=0}^\infty D_k \Indc(\tau > k).
	\end{align*}
	The meeting time $\tau$ is almost surely (a.s.) finite by the $\beta$-polynomial assumption, therefore $H_0(Y,Z)$ is the  limit of $H_0^n(Y,Z) := \sum_{k=0}^n D_k \Indc(\tau > k)$ in the  a.s. sense. We will now prove $H_0^n(Y,Z) \rightarrow H_0(Y,Z)$ in  $L^l$ , which further implies $\mathbb{E}[\lvert H_0(Y,Z)\rvert^l] <\infty$. 
	
	By the Minkowski’s inequality on the probability space $L^l(\Omega)$, we have
	\begin{align}
		\big(\bE[\lvert H_0^n(Y,Z) - H_0(Y,Z) \rvert^l]\big)^{1/l} &= \big(\bE[\big\lvert  \sum_{k=n+1}^\infty D_k \Indc(\tau > k)\big\rvert^l]\big)^{1/l} \\
		\label{eqn: Minkowski} & \leq \sum_{k=n+1}^\infty 	\big(\bE[\lvert D_k \Indc(\tau > k)\rvert^l]\big)^{1/l}.
	\end{align}

	Every term in \eqref{eqn: Minkowski} can be upper bounded by the H\"{o}lder's inequality
	\begin{align}
		\big(\bE[\lvert D_k \Indc(\tau > k)\rvert^l]\big)^{1/l} &\leq 	
		\big(\bE[\lvert D_k \rvert^p])^{1/p} \big(\bP(\tau > k)\big)^{1/q} \qquad{\text{here } 1/q = 1/l - 1/p} \\
		&\label{eqn: summable}\leq (2D_p)^{1/p} K_\beta^{1/q} k^{-\beta/q } \\&=  (2D_p)^{1/p} K_\beta^{1/q} k^{-\frac{\beta}{\frac 1l - \frac 1p} }.
	\end{align}
	Since $\beta > \frac 1l - \frac 1p > 0$, the right hand side of $\eqref{eqn: summable}$ is summable. Therefore we conclude 
	$$\sum_{k=n+1}^\infty 	\big(\bE[\lvert D_k \Indc(\tau > k)\rvert^l]\big)^{1/l} \rightarrow 0$$ as $n\rightarrow \infty$, so $H_0^n(Y,Z) \rightarrow H_0(Y,Z)$ in $L^l$. 
	
	In the second case,  exponential light tail implies $\beta$-polynomial tail for every $\beta > 0$, our result immediately follows from the first case. 
	
\end{proof}
The assumption $\bE[\lVert f(Y_t)\rVert^p] < D_p$ in Lemma \ref{lem: tail-moment} is generally satisfied as long as $f$ has $p$-th moment under the stationary distribution $\pi$. It remains to verify the tail conditions of $\tau$, i.e., formula \eqref{eqn: poly tail} or \eqref{eqn: expo tail}. The exponential tail \eqref{eqn: expo tail} and polynomial tail \eqref{eqn: poly tail} are closely related to the geometric ergodicity and polynomial  ergodicity of the underlying marginal Markov chain $P$, respectively. For simplicity, we only give conditions for the exponential tail here, which is provided in \cite{jacob2020unbiased}. The sufficient conditions of polynomial tail of $\tau$ can be founded in Theorem 2 of  \cite{middleton2020unbiased}.

\begin{prop}[Proposition 3.4 in \cite{jacob2020unbiased}]\label{prop:exponential tail}
	Suppose the Markov transition kernel described in Section \ref{subsubsec:unbiased MCMC} satisfies a geometric drift condition with a small set $\calS$ of the form $\calS =\{x: V(x)\leq L\}$ for $\lambda + b/(1+L) < 1$. Suppose there exists $\tilde \epsilon \in (0,1)$ such that
	\[
	\inf_{(x,y)\in \calS \times \calS} \barP((x,y), \calD) \geq \tilde\epsilon,
	\]
	where $\calD:= \{(x,x): x\in \Omega\}$ is the diagonal of $\Omega\times \Omega$. Then the meeting time $\tau$ has a exponential light tail. 
\end{prop}

Combining Lemma \ref{lem: tail-moment} and Proposition \ref{prop:exponential tail}, the proof of Proposition \ref{prop: moment, formal} is immediate.
\begin{proof}[Proof of Proposition \ref{prop: moment, formal}]
	By Proposition \ref{prop:exponential tail}, we know $\tau$ has an exponential  tail. Using the second case of Lemma \ref{lem: tail-moment}, our result follows.
\end{proof}
It is still possible to further strengthen Proposition \ref{prop: moment, formal} given extra assumptions on $\tau$ or $f$. For example, when $\tau$ has an exponential tail and $\bE_\pi[e^{\theta f}] < \infty$ for  a univariate $f$ and some $\theta>0$, one can then prove the JOA estimator also has an exponential moment, and thus has every finite-order moment. The existence of an exponential moment may help analyze the concentration properties of the JOA estimator.
\subsection{Other Technical Proofs}\label{subsec:other proof}
\subsubsection{Proof of Proposition \ref{prop:nested unbias}}\label{subsubsec: proof nested}
\begin{proof}
	Using the law of iterated expectation, the expectation of $\hat \lambda$ can be written as
	\begin{align*}
		\bE[\hat \lambda] &= \bE[\bE[\hat \lambda(x) | x] ]   \\
		& =  \int \bE[\hat \lambda(x)| x] \pi(dx)\\
		& =  \int  \lambda(x) \pi(dx)\\
		& = \bE_\pi[\lambda].
	\end{align*}
\end{proof}

\subsubsection{Proof of Proposition \ref{prop:delta transform}}
\begin{proof}
	We first show the unbiasedness of $\tilde H$. Notice that $\tilde H  =  H 1_{\lVert H\rVert\geq \delta} +  (H + 2\delta B\mathbf 1)1_{\lVert H\rVert < \delta}$ where $B\sim \sfU\{-1,1\}$ is independent with $H$. Therefore,
	\[
	\bE[\tilde H] = \bE[ H 1_{\lVert H\rVert\geq \delta}] +  \bE[(H + 2\delta B \mathbf 1)1_{\lVert H\rVert\geq \delta}] = \bE[ H 1_{H\geq \delta}] + \bE[H1_{H < \delta}] = \bE[H].
	\]
 For the variance, we can calculate:
 \begin{align*}
  \bE[\tilde H^2] = \bE[(H + 2\delta B1_{\lVert H\rVert\geq \delta})^2] &= \bE[H^2] + 4\delta^2\bE[1_{\lVert H\rVert\geq \delta}] + 4\delta \bE[H1_{\lVert H\rVert\geq \delta} B]\\
  & =  \bE[H^2] + 4\delta^2\bP[\lVert H\rVert\geq \delta],
 \end{align*}
 the last equality follows from the fact that $B$ has zero expectation and is independent with $H$. Finally, we have
 \begin{align*}
     \var[\tilde H] = \bE[\tilde H^2] - (\bE[\tilde H])^2 &= \bE[H^2] + 4\delta^2\bP[\lVert H\rVert\geq \delta] - \bE[H]^2\\
     &= \var[H] + 4\delta^2\bP[\lVert H\rVert\geq \delta]  \leq \var[H] + 4\delta^2,
 \end{align*}
 as desired.

\end{proof}

\subsubsection{Proof of Corollary \ref{cor:application_of_unbiasedness}}

\begin{proof}
	Let $W$ be the estimator output from Algorithm \ref{alg:MLMC}. Let $\mathsf{Cost}(W)$ denote its expected computational cost. From Theorem \ref{alg:MLMC}, we know both $\var(W)$ and $\mathsf{Cost}(W)$ is finite. For any fixed integer $n$, let $W_1, W_2,\ldots, W_n$ be the outputs of $n$ independent calls of Algorithm \ref{alg:MLMC}, and let $\tilde W := \frac{\sum_{i=1}^n W_i}{n}$ be its average. It follow from the unbiasedness of each $W_i$ that:
	\[
	\bE[(\tilde W - g(m(\pi)))^2]   =  \var(\tilde W) = \frac{\var(W)}{n}.
	\]
	Taking $n = \var(W)/\epsilon^2$, then the mean square error of $\tilde W$ will be no larger than $\epsilon^2$, and the expected computational cost will be $n \mathsf{Cost}(W) = \var(W)\mathsf{Cost}(W)/\epsilon^2 = \calO(1/\epsilon^2)$.
\end{proof}

\section{Extra numerical experiment for the Ising Model}\label{sec: extra experiment}
Let us denote the `natural statistics' of the Ising model by $h(\sigma) := -H(\sigma)$. In this example we are interested in estimating $1/\bE_\theta[h(\sigma)]$.  Standard calculation in exponential families yields:
\[
\frac{1}{\bE_\theta[h(\sigma)]} = \frac{1}{\log(Z(\theta))'} = \frac{Z(\theta)}{Z'(\theta)}.
\] 
Following the setups in \cite{jacob2020unbiased}, we set $n = 32$ (which means the sample space is of dimension $32^2 = 1024$) and use the JOA estimator for unbiased estimation of $\bE_\theta[h(\sigma)]$ by coupling two single-site Gibbs samplers, and feed these estimators as inputs for the unbiased MLMC estimator with parameter $p = 0.7, k = 10^5,$ and $m = 2\times 10^5$. We implement our estimator for a grid of $\theta$ values ranging from $0.23$ to $0.40$. For each $\theta$, we generate $10^5$ unbiased estimators and  report our results in Figure \ref{fig:naturalstat} below. Similar to the observations in \cite{jacob2020unbiased}, the meeting time increases exponentially as $\theta$ increases. Therefore it may be computationally demanding to generate unbiased estimators when $\theta$ is close to its critical temperature. Meanwhile, the standard deviation has an interesting $U$-shape pattern as $\theta$ increases, as shown in Figure \ref{fig:sub_sising2}. We have no idea how to explain this phenomenon theoretically. 
\begin{figure}[h]
	\begin{subfigure}{.5\textwidth}
		\centering
		\includegraphics[width=.9\linewidth]{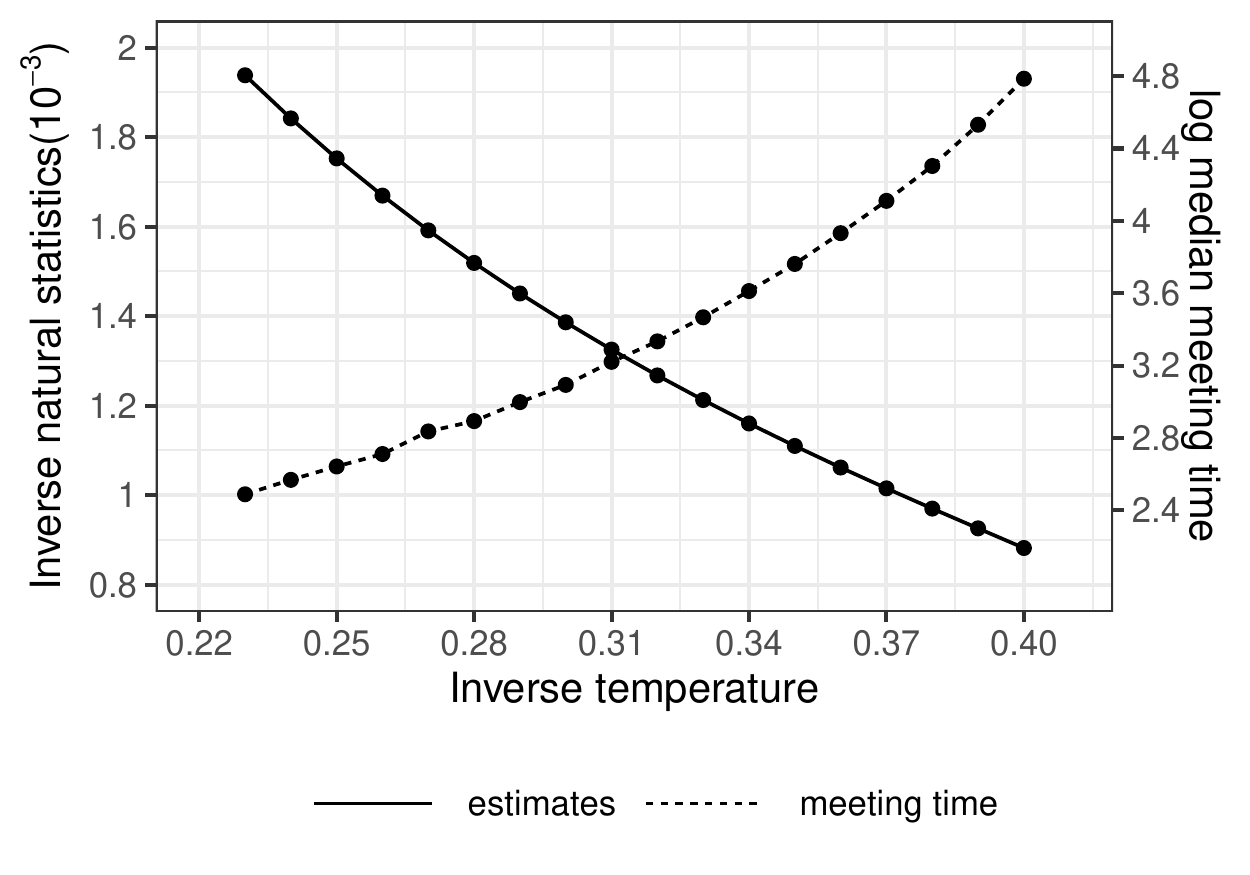}  
		\caption{Estimates and meeting times}
		\label{fig:sub_ising1}
	\end{subfigure}
	\begin{subfigure}{.5\textwidth}
		\centering
		\includegraphics[width=.9\linewidth]{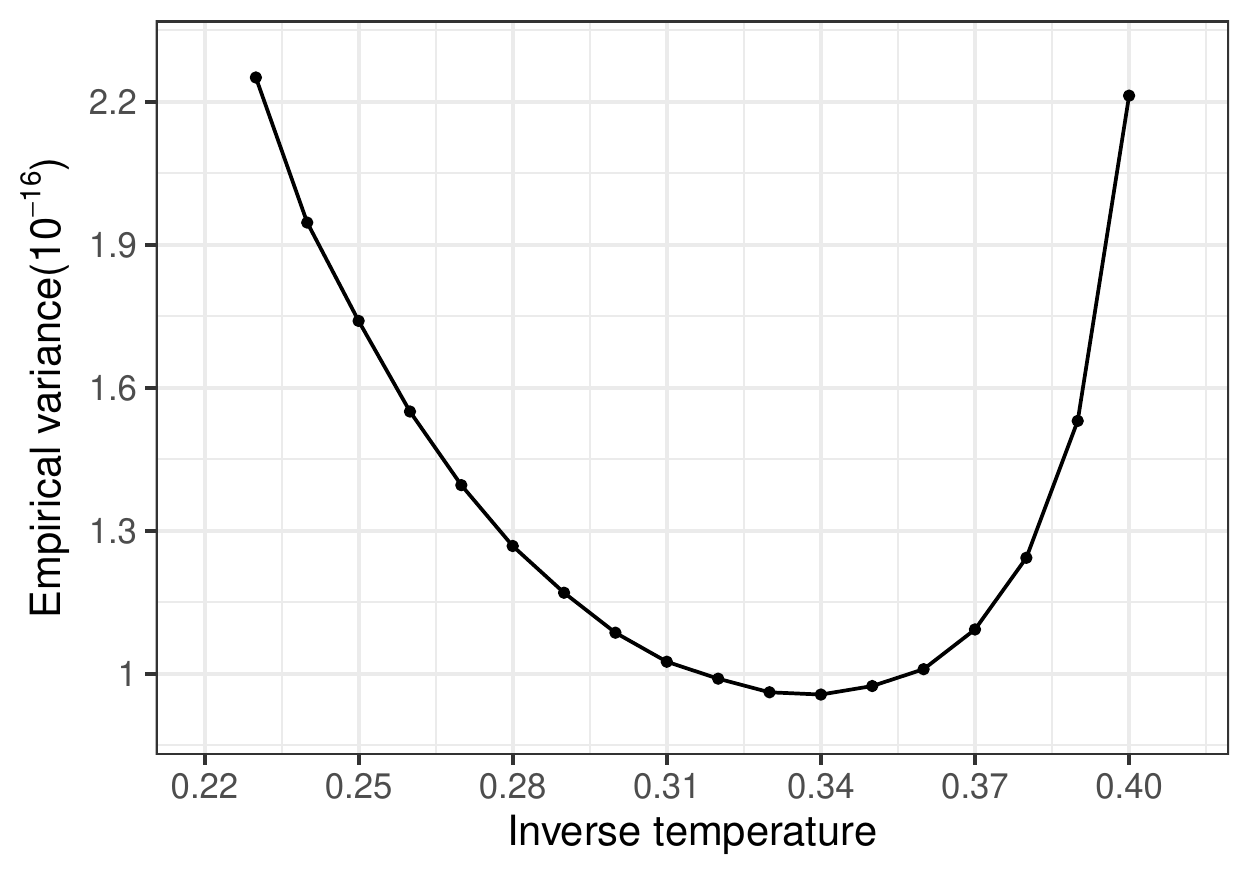}  
		\caption{Empirical variance}
		\label{fig:sub_sising2}
	\end{subfigure}
	\caption{Estimates, meeting times and standard deviations of $1/\bE_\theta[h(\sigma)]$ for $\theta\in\{0.23,0.24,\ldots, 0.4\}$. Left: the solid line stands for the empirical averages of $10^5$ unbiased estimators from Algorithm \ref{alg:MLMC}. The dashed line stands for the log median meeting time of the JOA estimators.}
	\label{fig:naturalstat}
\end{figure}

\end{document}